%% file: NIPS2025_WM.tex
\documentclass{article}

    \PassOptionsToPackage{numbers, compress}{natbib}

    \usepackage[final]{neurips_2025}

\usepackage[utf8]{inputenc} 
\usepackage[T1]{fontenc}    
\usepackage{hyperref}       
\usepackage{url}            
\usepackage{booktabs}       
\usepackage{amsfonts}       
\usepackage{nicefrac}       
\usepackage{microtype}      
\usepackage{xcolor}         

\input{preamble}

\usepackage{amsmath}
\usepackage{amssymb}
\usepackage{mathtools}
\usepackage{amsthm}
\usepackage{tikz}
\usepackage{pgf}
\usetikzlibrary{babel,automata,shapes}
\usepackage{fontawesome5}
\usetikzlibrary{shadows}
\usetikzlibrary{shadows.blur}

\usepackage[textsize=tiny]{todonotes}

\usepackage{url}

\usepackage{cases}
\usepackage{mdframed}
\usepackage[shortlabels]{enumitem}
\usepackage{extarrows}
\mathtoolsset{showmanualtags, showonlyrefs}
\usepackage{wrapfig}
\usepackage[table,x11names]{xcolor}
\usepackage{multirow}
\usepackage{scalerel}

\newcommand*\circled[1]{\tikz[baseline=(char.base)]{
            \node[shape=circle,draw,inner sep=0.3pt] (char) {#1};}}

\newcommand{\TV}{\mathsf{TV}}

\mathchardef\mhyphen="2D
\newcommand{\key}{\mathtt{key}}
\newcommand{\pt}{\mathrm{pt}}

\newcommand{\cP}{\mathcal{P}}

\newcommand{\cS}{\mathcal{S}}

\newcommand{\cV}{\mathcal{V}}

\newcommand{\cX}{\mathcal{X}}

\newcommand{\sD}{\mathsf{D}}

\title{Theoretically Grounded Framework for LLM Watermarking: A Distribution-Adaptive Approach}

\author{%
  Haiyun He\thanks{These authors contributed equally to this work.} \\
  HKUST (GZ)\\
  Guangzhou, China \\
  \texttt{haiyunhe@hkust-gz.edu.cn} \\
  \And
  Yepeng Liu$^*$ \\
  UC Santa Barbara \\
  Santa Barbara, CA, USA \\
  \texttt{yepengliu@ucsb.edu} \\
  \AND
  Ziqiao Wang \\
  Tongji University \\
  Shanghai, China \\
  \texttt{ziqiaowang@tongji.edu.cn} \\
  \And
  Yongyi Mao \\
  University of Ottawa \\
  Ottawa, ON, Canada \\
  \texttt{ymao@uottawa.ca} \\
  \And
  Yuheng Bu \\
  UC Santa Barbara \\
  Santa Barbara, CA, USA \\
  \texttt{buyuheng@ucsb.edu} \\
}

\begin{document}

\maketitle
\begin{abstract}
Watermarking has emerged as a crucial method to distinguish AI-generated text from human-created text. Current watermarking approaches often lack formal optimality guarantees or address the scheme and detector design separately. In this paper, we introduce a novel, unified theoretical framework for watermarking Large Language Models (LLMs) that jointly optimizes both the watermarking scheme and detector. Our approach aims to maximize detection performance while maintaining control over the worst-case false positive rate (FPR) and distortion on text quality. 
We derive closed-form optimal solutions for this joint design and characterize the fundamental trade-off between watermark detectability and distortion. Notably, we reveal that the optimal watermarking schemes should be adaptive to the LLM’s generative distribution.
Building on our theoretical insights, we propose a distortion-free, distribution-adaptive watermarking algorithm (DAWA) that leverages a surrogate model for model-agnosticism and efficiency. Experiments on Llama2-13B and Mistral-8$\times$7B models confirm the effectiveness of our approach, particularly at ultra-low FPRs. Our code is available at \url{https://github.com/yepengliu/DAWA}.


\end{abstract}

\section{Introduction}\label{Sec:intro}



Arising with Large Language Models (LLMs) \citep{touvron2023llama} is a double-edged sword: while they boost productivity, they also introduce new risks, including plagiarism, challenges to content accountability, and other forms of misuse.
Watermarking, a hidden and machine-verifiable tag inserted into LLMs' outputs, has therefore become a critical line of defense for publishers, educators, and regulators.


Existing watermarking techniques for AI-generated text are commonly grouped into two main categories: post-process and in-process \citep{liu2024survey, wu2025survey}. Post-process watermarking is applied after the text is generated \citep{brassil1995electronic, yoo2023robust, yang2023watermarking, munyer2023deeptextmark, yang2022tracing, sato2023embarrassingly, zhang2024remark, abdelnabi2021adversarial,martinian2005authentication, moulin2000information, chen2000design, 1564436, 4418489}, while in-process watermarking embeds watermarks  during generation \citep{wang2023towards, fairoze2023publicly, hu2023unbiased, huo2024token, zhang2023watermarks, tu2023waterbench, ren2023robust, hou2024semstamp,kirchenbauer2023watermark, zhao2023provable, liu2024adaptive, liu2024a, qu2024provably, fernandez2023three, fu2024gumbelsoft,wu2023dipmark,zhao2024permute, boroujeny2024multi, christ2024undetectable,giboulot2024watermax}.
There is a growing interest in in-process methods due to their invisibility, flexibility, and seamless integration with negligible latency. For additional related works, please refer to Appendix~\ref{App: other literature}.

However, realizing the full potential of in-process watermarking requires careful design. A key challenge is controlling the false positive rate (FPR), as even a single error can lead to serious consequences, such as wrongly accusing a human author. Maintaining a high true positive rate (TPR) while keeping an ultra-low FPR, e.g., $1\mathrm{e}{-05}$, is therefore essential. In addition, effective in-process watermarking should be both \textbf{detectable} and \textbf{distortion-free}: the watermark must be reliably identified under strict false positive rate (FPR) constraints, while preserving the quality and distribution of the original output \citep{christ2024undetectable, kuditipudi2023robust}.  
Moreover, a practical detector should be \textbf{model-agnostic}, operable without access to the original LLM or its prompt \citep{kirchenbauer2023watermark}.  Balancing all these goals is challenging, and existing methods often fall short in one or more dimensions.

Many heuristic designs typically embed watermarks into generated tokens by perturbing the token logits (e.g., the green-red list \citep{kirchenbauer2023watermark}) or modifying the sampling process (e.g.,  Gumbel-Max sampling \citep{gumbel2023}). Detection is usually performed using handcrafted score statistics.  However, these ``trial-and-error'' approaches rely heavily on empirical tuning, with no formal optimality guarantees. 

In principle, designing a watermarking system can be formulated as a constrained optimization problem:  maximizing the detection probability TPR with the FPR and text distortion under control. 
Recent theoretical efforts have taken steps towards this goal. For instance, \citet{takezawa2023necessary}, \citet{wouters2024optimizing}, and \citet{cai2024towards} focus on optimizing the logit perturbation strategy for the green-red list watermarking scheme. \citet{huang2023towards} frame watermarking as a statistical independence test between text and watermark and derive the optimal scheme for a \emph{fixed detector}, but stop short of a practical, model-agnostic algorithm. On the other hand, \citet{li2024statistical} optimizes the detector for a \emph{fixed watermarking scheme} using i.i.d.~pivotal statistics. 
Consequently, these approaches do not guarantee overall system optimality. A significant gap in these theoretical explorations is that none of them consider the \emph{joint optimization} of both the watermarking scheme and the detector.


Specifically, 
both the way the watermark is embedded and the type of signal used can be optimized. However, existing approaches often rely on fixed, simplistic designs, such as using randomly generated bits or samples from uniform distributions,
leaving much of the design space unexplored. These restrictive choices may prevent existing schemes from achieving optimal performance.

\begin{wrapfigure}[]{r}{.6\linewidth}
    \vspace{-2em}
    \begin{center}
    \includegraphics[width=1\linewidth]{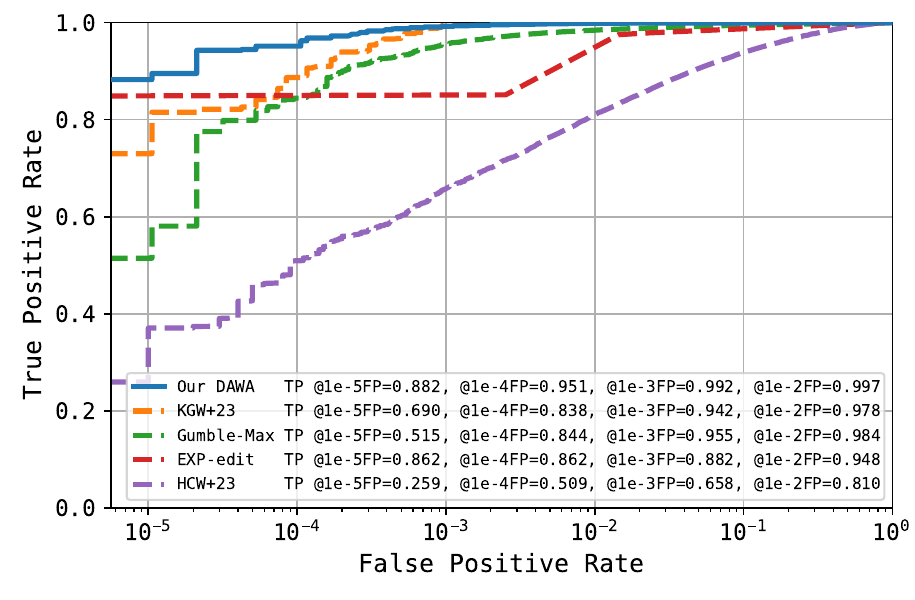}
    \vspace{-2em}
    \caption{Comparison of TPR at ultra-low FPR among different watermarking methods.}
    \label{fig:low_fpr}
    \vspace{-2em}
    \end{center}
\end{wrapfigure}

This paper aims to fill the gap. We develop a \textbf{novel, unified theoretical framework} that subsumes most existing in-process schemes, aiming to \emph{jointly optimize} the watermarking-detector pair for any token-sequence length $T$ that achieves the \emph{best trade-off} between watermark detectability and text distortion. Unlike the classical watermarking paradigm, which employs fixed watermark distributions, our framework generalizes watermarking to an \emph{adaptive setting} where the watermark signal exploits the LLM's generative distribution. This opens up one more degree of freedom to optimize the sampling distributions of watermark signals, thereby enhancing detection reliability at ultra-low FPRs, as shown in Figure \ref{fig:low_fpr}.

Our contributions can be summarized as follows:
\begin{itemize}
[leftmargin=*,topsep=0em,itemsep=0.1em]
    \item 
    In Section \ref{Sec: formulation}, we propose a unified theoretical framework for LLM watermarking and detection that encompasses most existing watermarking methods. This framework features a common randomness shared between watermark generation and detection to perform an independence test.
    
    \item In Section \ref{Sec: joint optimize}, we characterize the universally minimum Type-II error (i.e., ~$1-$TPR) as a function of FPR and text distortion level, revealing a fundamental trade-off between \textbf{detectability and distortion}.  
    More importantly,  we identify the closed-form jointly optimal solutions for watermarking schemes and detectors, providing 
    a guideline for practical design,
    i.e., watermarking schemes should adapt to the 
    generative distribution of LLMs.
    \item 
    In Section \ref{section: token-level}, we translate the sequence-level optimum to a practical token-level watermarking design. We prove that it retains reliable detection performance and is inherently robust against token replacement attacks. 
    In Section \ref{Sec: algorithm}, we introduce  \textbf{DAWA} (\textbf{D}istribution-\textbf{A}daptive \textbf{W}atermarking \textbf{A}lgorithm), a distortion-free implementation of the token-level design,  which leverages a surrogate language model and the Gumbel-Max sampling trick to achieve \textbf{model-agnosticism} and computational efficiency.
    
    \item In Section \ref{Sec: experiments}, we conduct extensive experiments on Llama2-13B \citep{touvron2023llama} and Mistral-8$\times$7B \citep{jiang2023mistral}, across multiple datasets. DAWA consistently outperforms the compared methods, even under token replacement attacks, and maintains high text quality. As shown in Figure \ref{fig:low_fpr}, DAWA achieves superior detection capabilities at ultra-low FPRs.
    \item Lastly, we sketch in Appendix \ref{App: ext to robust} how to extend our theoretical framework to semantic‑invariant watermarking removal attacks and derive the associated detectability--distortion--robustness trade-off, guiding future semantic‑based watermark designs.
\end{itemize}



\section{Preliminaries and Problem Formulation}\label{Sec: formulation}

\textbf{Notations.} 
For a sequence of random variables $X_1, \ldots, X_n$, and any $i, j \in [n]$ with  $i\le j$, we denote $X_i^j\coloneqq (X_i,\ldots,X_j)$. 
We may use distortion function, namely, a function $\sD: \calP(\cX) \times \calP(\cX) \to [0,+\infty)$
to measure the dissimilarity between two distributions in $\calP(\cX)$. 
For example, 
the total variation distance, as a distortion, between $\mu,\nu\in\cP(\cX)$ is $\sD_\TV(\mu,\nu)\coloneqq\int \frac{1}{2}| \frac{d  \mu}{ d  \nu}-1| \,  d \nu$. 
For any set $A\subseteq \cX$, we use $\delta_A$ to denote its indicator function, namely, $\delta_A(x)\coloneqq \mathbbm{1}\{x\in A\}$. Additionally, we denote $(x)_+:=\max\{x, 0\}$ and $x\wedge y:= \min\{x, y\}$.


\textbf{Tokenization and NTP.}
LLMs process text through ``tokenization,'' namely, breaking it down into words or word fragments called ``tokens.'' An LLM generates text token by token.
Let $\calV$ denote the token vocabulary, typically of size $|\calV|=\calO(10^4)$ \citep{liu2019roberta,radford2019language,zhang2022opt,touvron2023llama}.
An \emph{unwatermarked} LLM generates the next token $X_t$ based on a prompt $\pt$ and the previous tokens $x_1^{t-1}$
by sampling the Next-Token Prediction (NTP) distribution $Q_{X_t|x_1^{t-1},\pt}$. For simplicity, the prompt dependency is suppressed in notation throughout the paper. The joint distribution of a length-$T$ generated token sequence $X_1^T$ is then given by $Q_{X_1^{T}}\coloneqq\prod_{t=1}^T Q_{X_t|X_1^{t-1}}$, which we assume to be identical to one that governs the human-generated text.

\begin{wrapfigure}[]{r}{.6\linewidth}
\vspace{-1.5em}
\centering
\begin{tikzpicture}[scale=0.73,every node/.style={transform shape}]



    \node (start) at (-1, 4.5) {};

    \tikzset{
        mybox/.style={
            rectangle, 
            draw, 
            fill=white, 
            minimum height=1cm, 
            minimum width=2cm, 
            align=center,
            thick,
            drop shadow
        },
        myboxtall/.style={
            rectangle, 
            draw, 
            fill=white, 
            minimum height=1.5cm, 
            minimum width=2cm, 
            align=center,
            thick,
            drop shadow
        },
        cloudnode/.style={cloud, draw, fill=white, cloud puffs=10, cloud puff arc=120, aspect=4, inner ysep=1em, minimum width=1cm, minimum height=.5cm},
        roundnode/.style={ellipse, draw, fill=white, minimum width=2cm, minimum height=1cm}
    }

        \node[myboxtall] (sampler1) at (9, 4.5) {Sampler};
        \node[mybox] (WM) at (5, 4) {Watermarking \\ Scheme};
        \node[myboxtall, below = 2 of WM] (dec) {Detector};
        \node[] (common) at (5,2.5) {Common Randomness $\zeta_1^T$};

        \node[gray] at (0.5,2.8) {Watermark Generation};
        \draw[-,dashed, thick, gray] (-1.2,2.5)--++(3.5,0);
        \node[gray] at (0.5,2.3) {Watermark Detection};

        \node[myboxtall] (LLM) at (0.5,4.5) {\Large \faDesktop \\ LLM};
        \draw[->, dashed, very thick] (LLM.east)++(0, 0.5) -- (8, 5);
        \node[] at (7,5.3) {$Q_{X_1^T}$};
        \draw[->, dashed, very thick] (sampler1.south)++(0.5,0) -- ++(0, -3.5) -- ++(-3.5, 0) node[midway, above] {$x_1^T$} node[midway, below] {unwatermarked text};

        \draw[->, very thick, color = SteelBlue2] (LLM.east)++(0, -0.5) -- (WM.west) node[near end, above] {$Q_{X_1^T}$};
        \draw[->, very thick, color = SteelBlue2] (WM.east) -- (8,4) node[midway, above] {$P_{X_1^T|\zeta_1^T}$};
        \draw[->, very thick, color = SteelBlue2] (sampler1.south)++(-.5,0) -- ++(0,-2.5) -- ++(-2.5,0) node[midway, above] {$x_1^T$} node[midway, below] {watermarked text};

        \draw[->, densely dotted, very thick, color = gray] (common.north) -- (WM.south);
        \draw[->, densely dotted, very thick, color=gray] (common.south) -- (dec.north);
        
        \draw[->, very thick, color = SteelBlue2] (dec.west)++(0, 0.5) -- ++(-2.5,0) node[midway, above] {1};
        \draw[->, dashed, very thick] (dec.west)++(0, -0.5) -- ++(-2.5,0) node[midway, above] {0};
        \node at (0.3,1.2) {Watermarked};
        \node at (0.3,.2) {Unwatermarked};

\end{tikzpicture}
    \caption{Overview of LLM watermarking and detection.}
    \label{Fig: system}
    \vspace{-1em}
\end{wrapfigure}
\textbf{A Framework for Watermarking Scheme.} 
Traditional post-hoc detectors identify AI-generated text by dividing the entire text space into rejection and acceptance regions, which relies on the assumption that certain sentences are unlikely to be written by humans. In contrast, modern LLM watermarking schemes achieve the same goal by analyzing the dependence structure between text $X_1^T$ and an auxiliary random sequence $\zeta_1^T$, thereby avoiding this restriction. 


In this paper, we propose a general framework for LLM watermarking and detection, as shown in Figure \ref{Fig: system}, which encompasses most of the existing watermarking schemes. The watermarking scheme and detector share a common randomness represented by an \emph{auxiliary random sequence} $\zeta_1^T$ drawn from a space $\calZ^T$ (either discrete or continuous). After passing through a watermarking scheme, the watermarked LLM samples token sequence according to the modified NTP distribution $P_{X_t|x_1^{t-1},\zeta_1^T}$, where $P_{X_1^T|\zeta_1^T}=\prod_{t=1}^T P_{X_t|X_1^{t-1},\zeta_1^T}$. This process associates the generated text $X_1^T$ with an auxiliary sequence $\zeta_1^T$. Thus, the joint distribution of the watermarked token sequence $X_1^T$ is $P_{X_1^T}$, which might be different from the original $Q_{X_1^T}$. 
The detector can then distinguish whether the received sequence $X_1^T$ is watermarked or not based on the common randomness.




To evaluate the \emph{distortion level} of a watermarking scheme, we measure the statistical divergence between the watermarked text distribution $P_{X_1^T}$ and the original one $Q_{X_1^T}$.
\begin{definition}[$\epsilon$-distorted watermarking scheme]
    A watermarking scheme is \emph{$\epsilon$-distorted} with respect to distortion $\sD$, if $\sD(P_{X_1^T}, Q_{X_1^T})\leq \epsilon$. Here, $\sD$ can be any distortion metric.
\end{definition}
 Common examples of such divergences include squared distance, total variation, KL divergence, and Wasserstein distance. For $\epsilon=0$, the watermarking scheme is \emph{distortion-free}.

Specifically, our formulation allows the auxiliary random sequence $\zeta_1^T$ to take an arbitrary structure, which contrasts the rather restricted i.i.d.\ assumption considered in \citet[Working Hypothesis 2.1]{li2024statistical}. In practice, $\zeta_1^T$ is usually randomly generated using a shared $\key$ accessible during both watermark generation and detection.
At first glance, our formulation may appear abstract, but its flexibility enables existing watermarking schemes to be interpreted as special cases within this framework.


\begin{example}[Existing watermarking schemes as special cases]\label{Ex: watermark schemes}
    In the Green-Red List watermarking scheme \citep{kirchenbauer2023watermark}, at each position $t$, the vocabulary $\calV$ is randomly split into a green list $\calG$ and a red list $\calR$, with $|\calG|=\rho|\calV|$ for some $\rho\in(0,1)$. This split is represented by a $|\calV|$-dimensional binary auxiliary variable $\zeta_t$, indexed by $x \in \calV$, where $\zeta_t(x)=1$ means $x\in\calG$; otherwise, $x\in\calR$.  
    The watermarking scheme is as follows:
    \begin{enumerate}[noitemsep,topsep=0pt,parsep=0pt,partopsep=0pt, leftmargin=*,label=--]
        \item Compute a hash of the previous token $X_{t-1}$ using a hash function $h:\calV\times \bbR\to\bbR$ and a shared secret $\key$, i.e., $h(X_{t-1},\key)$.
        \item Use $h(X_{t-1},\key)$ as a seed to uniformly sample the auxiliary variable $\zeta_t$ from the set $\{\zeta\in\{0,1\}^{|\calV|}: \|\zeta\|_1=\rho|\calV|\}$ to construct the green list $\calG$.
        \item Sample $X_t$ from the adjusted NTP distribution which increases the logit of tokens in $\calG$ by $\delta>0$:
    \setlength{\abovedisplayskip}{5pt}
    \setlength{\belowdisplayskip}{-5pt}
    \setlength{\abovedisplayshortskip}{5pt}
    \setlength{\belowdisplayshortskip}{-5pt}
    \begin{equation}
     \resizebox{.5\hsize}{!}{$
        P_{X_t|x_1^{t-1},\zeta_t}(x)= \frac{Q_{X_t|x_1^{t-1}}(x)\exp(\delta\cdot\mathbbm{1}\{\zeta_t(x)=1\})}{\sum_{x\in\calV}Q_{X_t|x_1^{t-1}}(x)\exp(\delta\cdot\mathbbm{1}\{\zeta_t(x)=1\})}.
         $}
    \end{equation}
    \end{enumerate}
    
    How our formulation encompasses several other watermarking schemes is provided in Appendix \ref{App: existing watermark}.
\end{example}




\textbf{Hypothesis Testing for Watermark Detection.}
Note that a sequence $X_1^T$ generated by a watermarked LLM depends on $\zeta_1^T$, while $X_1^T$ and $\zeta_1^T$ are independent if written by humans.
Therefore, detection involves distinguishing the following two hypotheses based on the pair $(X_1^T,\zeta_1^T)$:

\begin{itemize}
[leftmargin=.2cm,topsep=-0cm, parsep=-0cm,labelsep=-0.5cm,align=left]
    \item $\rmH_0$: $X_1^T$ is generated by a human, i.e., $(X_1^T,\zeta_1^T)\sim Q_{X_1^T}\otimes P_{\zeta_1^T}$;
    \item $\rmH_1$: $X_1^T$ is generated by a watermarked LLM, i.e., $(X_1^T,\zeta_1^T)\sim P_{X_1^T,\zeta_1^T}$.
\end{itemize}
We consider a model-agnostic detector $\gamma: \calV^{T}\times \calZ^T \to \{0,1\}$, which maps $(X_1^T,\zeta_1^T)$ to the hypothesis index (see Figure \ref{Fig: system}). 
In theory, we assume that the auxiliary sequence $\zeta_1^T$ can be fully recovered from $X_1^T$ and the common randomness,
\emph{while this assumption is dropped in practice. }

Detection performance is measured by the Type-I (false positive) and Type-II (false negative) errors: 
\begin{align}
    \text{Type-I (i.e., FPR): }\quad \beta_0(\gamma,Q_{X_1^T},P_{\zeta_1^T})\!&\coloneqq\Pr(\gamma(X_1^T,\zeta_1^T)\ne 0 \mid \rmH_0),\\
     \text{Type-II (i.e., $1-$TPR): } \quad \beta_1(\gamma,P_{X_1^T,\zeta_1^T} )\!&\coloneqq\Pr(\gamma(X_1^T,\zeta_1^T)\ne 1 \mid \rmH_1).
    \label{Def: type I and II}
\end{align}
\noindent\textbf{Optimization Problem.}  
Given that human-generated texts can vary widely, within our proposed framework, we aim to control the \emph{worst-case} Type-I error $\sup_{Q_{X_1^T}}\beta_0(\gamma,Q_{X_1^T},P_{\zeta_1^T})$ at a given $\alpha\in(0,1)$ while minimizing Type-II error. Our objective is to design an $\epsilon$-distorted watermarking scheme and a model-agnostic detector by solving the following optimization:
\setlength{\abovedisplayskip}{5pt}
\setlength{\belowdisplayskip}{-3pt}
\setlength{\abovedisplayshortskip}{5pt}
\setlength{\belowdisplayshortskip}{-3pt}
\begin{align}
   &\inf_{\gamma,P_{X_1^T,\zeta_1^T}} \beta_1(\gamma,P_{X_1^T,\zeta_1^T} ) \tag{Opt-O} \label{Eq: opt-O} \quad \text{ s.t. } \sup_{Q_{X_1^T}}\beta_0(\gamma,Q_{X_1^T},P_{\zeta_1^T})\leq \alpha,\; \sD(P_{X_1^{T}}, Q_{X_1^{T}})\leq \epsilon.  
\end{align}
The optimal objective value, denoted as $\beta_1^*(Q_{X_1^{T}}, \alpha, \epsilon)$, is termed as \emph{universally minimum Type-II error}. This universality is due to its applicability across all potential detectors and watermarking schemes, as well as its validity under the worst-case Type-I error scenario.



\section{Jointly Optimal Watermarking Scheme and Detector}\label{Sec: joint optimize}
In this section, we aim to solve the optimization in \eqref{Eq: opt-O} and identify the jointly optimal watermarking scheme and detector. However, solving \eqref{Eq: opt-O} is challenging due to the binary nature of $\gamma$ and the vast set of possible $\gamma$, sized $2^{|\calV|^T|\calZ|^T}$. To address this, we begin with a specific $\gamma(X_1^T,\zeta_1^T)=\mathbbm{1}\{(X_1^T,\zeta_1^T)\in\calA_1\}$, where $\calA_1$ defines the acceptance region for $\rmH_1$, aiming to uncover a potential structure for the optimal detector.
To this end, we  simplify \eqref{Eq: opt-O} as
\setlength{\abovedisplayskip}{5pt}
\setlength{\belowdisplayskip}{-3pt}
\setlength{\abovedisplayshortskip}{5pt}
\setlength{\belowdisplayshortskip}{-3pt}
\begin{align}
&\inf_{P_{X_1^T,\zeta_1^T}}\beta_1(\gamma,P_{X_1^T,\zeta_1^T} ) \tag{Opt-I}\label{Eq: opt-I}\quad \text{s.t. } \sup_{Q_{X_1^T}} \beta_0(\gamma,Q_{X_1^T},P_{\zeta_1^T}) \leq \alpha, \; \sD(P_{X_1^T}, Q_{X_1^T}) \leq \epsilon. 
\end{align}

\paragraph{Error-Distortion Tradeoff.} We first derive a lower bound for the minimum Type-II error in \eqref{Eq: opt-I}, which surprisingly does not depend on the selected detector 
$\gamma$ and therefore also applies to \eqref{Eq: opt-O}. We then pinpoint a type of detector and watermarking scheme that attains this lower bound, indicating that it represents the universally minimum Type-II error. Thus, the proposed detector and watermarking scheme are jointly optimal, as detailed in Theorem~\ref{Thm: optimal detector}.
The theorem below establishes this universal minimum Type-II error for all feasible watermarking schemes and detectors.


\vspace{5pt}
\begin{theorem}[Universally minimum Type-II error]\label{Thm: universal min typeII}
    The universally minimum Type-II error attained from \eqref{Eq: opt-O} is 
    \begin{small}
    \begin{equation}
     \beta_1^*(Q_{X_1^T},\alpha, \epsilon)= \!\!\! \min_{\substack{P_{X_1^T}:  \sD(P_{X_1^T}, Q_{X_1^T} )\leq \epsilon}}\sum\nolimits_{x_1^T}(P_{X_1^T}(x_1^T)- \alpha)_+,  \label{Eq: Type-II LB}
    \end{equation}
    \end{small}
    which is achieved by the  watermarked  distribution
    \begin{small}
    \begin{equation}
    P_{X_1^T}^*=\argmin_{P_{X_1^T}:\sD(P_{X_1^T}, Q_{X_1^T} )\leq \epsilon}\sum\nolimits_{x_1^T}(P_{X_1^T}(x_1^T)- \alpha)_+.  \label{Eq: opt X marginal}
    \end{equation}
    \end{small}
    By setting $\sD$ as total variation distance $\sD_\TV$,  \eqref{Eq: Type-II LB} can be simplified as follows: 
    \begin{small}
    \begin{equation}
        \beta_1^*(Q_{X_1^T},\alpha, \epsilon)=\Big(\sum\nolimits_{x_1^T}(Q_{X_1^T}(x_1^T)- \alpha)_+ -\epsilon \Big)_+, \quad \text{if}\;  \sum\nolimits_{x_1^T}(\alpha-Q_{X_1^T}(x_1^T))_+\geq \epsilon.
    \end{equation}
    \end{small}
\end{theorem}

\begin{wrapfigure}[10]{r}{.35\linewidth}
    \vspace{-25pt}
    \begin{center}
    \includegraphics[trim={1.1cm .4cm .4cm 0}, width=.8\linewidth]{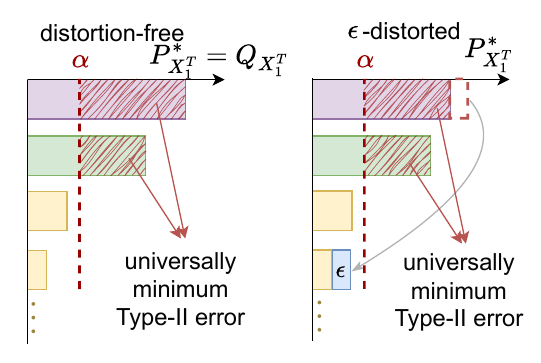}
    \caption{Illustration of error--distortion trade-off.}
    \label{Fig: marginal}
    \end{center}
\end{wrapfigure}
The proof of Theorem \ref{Thm: universal min typeII} is deferred to Appendix \ref{App: pf of Thm: universal min typeII}. 
Theorem \ref{Thm: universal min typeII} shows that, for any watermarking scheme, the fundamental limits of detection performance depend on the original NTP distribution of the LLM. When the original $Q_{X_1^T}$ is more concentrated (low entropy), the minimum achievable detection error increases. This hints that it is inherently difficult to watermark low-entropy text. However, increasing the allowable distortion 
$\epsilon$ can enhance the capacity for reducing detection errors, as illustrated in Figure \ref{Fig: marginal}.
Moreover, we find that $\beta_1^*(Q_{X_1^T},\alpha,\epsilon)$ matches the minimum Type-II error from \citet[Theorem 3.2]{huang2023towards}, which is notably optimal for their specific detector. Our results, however, establish that this is the universally minimum Type-II error across all possible detectors and watermarking schemes, indicating that their detector belongs to the set of optimal detectors described below.


\paragraph{Jointly Optimal Design.}  
We now present the jointly optimal watermarking schemes and detectors that achieve the universally minimum Type-II error in Theorem \ref{Thm: universal min typeII}, i.e., the solution to \eqref{Eq: opt-O}. 
The key takeaway in the optimal design is:
(i) for any given LLM $Q_{X_1^T}$, we can find a valid bijective function $g$ (not unique) to construct a jointly optimal pair of watermarking scheme and detector; (ii) as a result of maximizing the detection performance of dependency against independence between $X_1^T$ and $\zeta_1^T$, the optimal watermarking scheme $P_{\zeta_1^T|X_1^T}^*$ turns out to be a nearly deterministic. 


\begin{theorem}[(Informal Statement) Jointly optimal watermarking schemes and detectors] \label{Thm: optimal detector}
%
 The class of optimal detectors is given by
\setlength{\abovedisplayskip}{0pt}
\setlength{\belowdisplayskip}{2pt}
\setlength{\abovedisplayshortskip}{0pt}
\setlength{\belowdisplayshortskip}{2pt}
\begin{equation}
        \!\!\!\gamma^*\in \Gamma^* \!\coloneqq \! \big\{\gamma\mid \gamma(X_1^T,\zeta_1^T)= \mathbbm{1}\{X_1^T=g(\zeta_1^T)\},\text{ for some bijective } g:\calZ^T\to \cS \supset \cV^T \big\},
    \label{Eq: opt detector class}
\end{equation}
where $\calZ^T=g^{-1}(\calV^T)\cup \{\tilde{\zeta}_1^T\}$ and $\tilde{\zeta}_1^T$ is a redundant auxiliary sequence that is not a preimage of any token sequence for $g$. Given any $(Q_{X_1^T},\epsilon)$,
the corresponding optimal watermarking scheme takes the form $P_{X_1^T,\zeta_1^T}^*=P_{X_1^T}^*P_{\zeta_1^T|X_1^T}^*$, where  $P_{X_1^T}^*$ (c.f.~\eqref{Eq: opt X marginal}) depends on $(Q_{X_1^T},\epsilon)$, and the mapping $P_{\zeta_1^T|X_1^T}^*$  depends on the chosen detector $\gamma^*$. Full details are provided in Appendix \ref{App: pf of Thm: optimal detector}.

\end{theorem}
In Appendix \ref{App: pf of Thm: optimal detector}, a formal and general statement of Theorem \ref{Thm: optimal detector} shows that the optimal class of detectors extends to any valid surjective function $g$ with a different input space. To illustrate this with a simple example, suppose $\calV^T=\{a,b,c\}$. We can define an extended set $\calS=\{a,b,c,\# \}$ and construct a bijective mapping $g$ from an  auxiliary set $\calZ^T=\{1,2,3,4\}$ to  $\calS$.
Such a class $\Gamma^*$ is then universally optimal, meaning that to guarantee the construction of a watermarking scheme that maximizes the detection performance, the detector must be chosen from $\Gamma^*$.

\textbf{Discussions on Theoretically Optimal Watermarking Scheme.} 
A detailed illustration of the optimal watermarking scheme is provided in Appendix \ref{App: toy example of opt watermark}, with several important remarks as follows. 

\textbf{First}, we observe that the derived optimal watermarking scheme  $P_{X_1^T,\zeta_1^T}^*$   for any  $\gamma^*\in\Gamma^*$  is \emph{adaptive} to the original LLM output distribution  $Q_{X_1^T}$. This observation suggests that, to maximize watermark detection performance, watermarking schemes should fully leverage generative modeling and make the sampling of auxiliary sequence adaptive to $Q_{X_1^T}$. 
This approach contrasts with existing watermarking schemes, which typically sample the auxiliary sequence according to a given uniform distribution, without adapting it to the LLM NTP distribution.
This insight serves as a foundation for the design of our practical watermarking scheme, which will be introduced in  Section \ref{section: token-level}.


\textbf{Second}, in order to control the worst-case FPR, the construction of $P_{X_1^T,\zeta_1^T}^*$ relies on the redundant auxiliary sequence $\tilde{\zeta}_1^T$ included in the auxiliary alphabet $\calZ^T$, which satisfies  $\gamma^*(x_1^T,\tilde{\zeta}_1^T)=0$ for all $x_1^T$. This sequence $\tilde{\zeta}_1^T$ plays a critical role in our proposed algorithm. 
Specifically, if $P_{X_1^T}^*(x_1^T)>\alpha$ (indicating a low-entropy text, e.g., a celebrity's name),  it may be mapped to the redundant $\tilde{\zeta}_1^T$, making it harder to detect as watermarked. Thus, the optimal watermarking scheme $P_{X_1^T,\zeta_1^T}^*$ is particularly effective in reducing the FPR for low-entropy texts. 



\paragraph{Practical Challenges.} Given the theoretically optimal structure, there are still a few practical challenges in its direct implementation.   {\bf \circled{1}} Designing a proper function $g$, an alphabet $\calZ^T$ and the corresponding $P_{\scaleto{X_1^T,\zeta_1^T}{7pt}}^*$  is challenging, as $|\calV|^T$ grows exponentially with  $T$, making it hard to identify all pairs  $(x_1^T, \zeta_1^T)$  such that $ x_1^T = g(\zeta_1^T)$.  {\bf \circled{2}} The optimality is derived for static scenarios with a fixed token length $T$, making it unsuitable for dynamic scenarios where the tokens are generated incrementally with varying $T$.  {\bf \circled{3}}  In the theoretical analysis, we assume full recovery of the auxiliary sequence $\zeta_1^T$ during detection. However, in practice, the detector only receives the token sequence $X_1^T$, and reconstructing the auxiliary sequence $\zeta_1^T$ from $X_1^T$ poses a challenge.

These practical constraints motivate the development of a more feasible version of the theoretically optimal scheme. In Section \ref{section: token-level}, we adapt it to a practical token-level optimal scheme to address {\bf \circled{1}}  and  {\bf \circled{2}}; in Section \ref{Sec: algorithm}, we implement the token-level design with a novel algorithm utilizing a surrogate language model and the Gumbel-Max trick \cite{gumbel1954statistical} to overcome {\bf \circled{3}}.

\section{Practical Token-level Optimal Watermarking Design}
\label{section: token-level}


In this section, we present a practical approach that approximates the theoretical framework while ensuring its applicability to real-world scenarios. Building on the fixed-length optimal scheme, we naturally extend it to accommodate varying-length scenarios by constructing the optimal watermarking scheme incrementally for each token, i.e., $P^*_{X_t,\zeta_t|x_1^{t-1},\zeta_1^{t-1}}$ for all $t=1,2,\ldots$.

To lay the groundwork, we first revisit heuristic detectors for some existing watermarking schemes.

\begin{example}[Examples of heuristic detectors]
 Two example detectors from existing works: 
    \begin{itemize}[leftmargin=*,noitemsep,topsep=-5pt]
        \item Green-Red List watermark detector \citep{kirchenbauer2023watermark}: $\gamma(X_1^T,\zeta_1^T)=\mathbbm{1} \{\frac{2}{\sqrt{T}}(\sum_{t=1}^T \mathbbm{1}\{\zeta_t(X_t)=1\}-\rho T) \geq \lambda\}$ where $\lambda>0$, $\rho \in (0,1)$,  and $\zeta_t=(\zeta_t(x))_{x\in\calV}$ is uniformly sampled from $\{\zeta\!\in\!\{0,1\}^{|\calV|}\!: \!\|\zeta\|_1\!=\!\rho|\calV|\}$ with the seed $\mathrm{hash}(\!X_{t-1},\key\!)$. 

        \item Gumbel-Max watermark detector \citep{gumbel2023}: $\gamma(X_1^T,\zeta_1^T)=\mathbbm{1} \{-\sum_{t=1}^T\log(1-\zeta_t(X_t))\geq \lambda\} \}$ where $\lambda>0$, and $\zeta_t=(\zeta_t(x))_{x\in\calV}$ is uniformly sampled from $[0,1]^{|\calV|}$ with the seed $\mathrm{hash}(X_{t-1}^{t-n},\key)$ for some $n$.
    \end{itemize}
\end{example}


\textbf{Practical Detector Design.}
We observe that the commonly used heuristic detectors take the non-optimal form by averaging the test statistics over each token: $\gamma(X_1^T,\zeta_1^T)=\mathbbm{1} \big\{\frac{1}{T}\sum_{t=1}^T\text{Test Statistics of }(X_t,\zeta_t) \geq \lambda \big\}$.
This token-level design provides several advantages: (1) incremental computation of detectors for any 
 $T$ and 2) token-level watermarking with the alphabet depending only on the fixed size $ |\calV| $. 
Inspired by these detectors, we propose the following detector to address the issues {\bf \circled{1}}  and  {\bf \circled{2}}  mentioned earlier: 
\vspace{2pt}
\setlength{\abovedisplayskip}{1pt}
\setlength{\belowdisplayskip}{0pt}
\setlength{\abovedisplayshortskip}{1pt}
\setlength{\belowdisplayshortskip}{0pt}
\begin{equation}
    \gamma_{\mathrm{tk}}(X_1^T,\zeta_1^T)=\mathbbm{1} \Bigg\{\frac{1}{T}\sum_{t=1}^T\underbrace{\mathbbm{1}\{X_t=g_{\mathrm{tk}}(\zeta_t)\}}_{\text{Token-level adaptation of \eqref{Eq: opt detector class}}} \geq \lambda \Bigg\}, \label{Eq: token level detector}
\end{equation}
\vspace{1pt}
for some bijective function $g_{\mathrm{tk}}:\calZ\to\calS \supset \calV$, where $\calZ=g_{\rm{tk}}^{-1}(\calV)\cup\{\tilde{\zeta}\}$ for some redundant auxiliary value $\tilde{\zeta}$ not being the preimage of any token $x\in\calV$ for $g$.
This detector combines the advantages of existing token-level detectors with the optimal design from Theorem \ref{Thm: optimal detector}. The test statistic for each token  $(X_t, \zeta_t)$  is optimal at position $t$, enabling a token-level optimal watermarking scheme that improves the detection performance for each token. 

\textbf{Token-Level Optimal Watermarking Scheme.} Following the same rule in Theorem \ref{Thm: optimal detector} and Appendix \ref{App: pf of Thm: optimal detector}, the token-level optimal watermarking scheme is \emph{sequentially} constructed  based on $\mathbbm{1}\{X_t=g_{\mathrm{tk}}(\zeta_t)\}$ in \eqref{Eq: token level detector} and the NTP distribution at each position $t$, acting only on the token vocabulary $\calV$. This approach addresses the challenges  {\bf \circled{1}} and  {\bf \circled{2}} as well.
Notably, the resulting distribution of the token-level optimal scheme for the auxiliary variable $\zeta_t$ is adaptive to the original NTP distribution $Q_{X_t|x_1^{t-1}}$. Moreover, the resulting distribution on $X_t$ is given by (comparable to $P_{X_1^T}^*$ in Theorem \ref{Thm: optimal detector}) 

\setlength{\abovedisplayskip}{-5pt}
\setlength{\belowdisplayskip}{-10pt}
\setlength{\abovedisplayshortskip}{-5pt}
\setlength{\belowdisplayshortskip}{-10pt}
\begin{small}
\begin{equation}
P_{X_t|x_1^{t-1}}^* \coloneqq \!\! \argmin_{\substack{\scaleto{ P_{X_t|x_1^{t-1}}:\sD(P_{X_t|x_1^{t-1}}, Q_{X_t|x_1^{t-1}})\leq \epsilon}{8pt}}}\sum_{x\in\calV}(P_{X_t|x_1^{t-1}}(x)- \eta)_+, \label{Eq: token opt P_X}
\end{equation}
\end{small}
\noindent where $\eta\in(0,1)$ is the \emph{token-level FPR constraint}, which is typically much greater than the sequence-level FPR constraint $\alpha$. With a proper choice of $\eta$, we can effectively control $\alpha$.
Under this scheme, we add watermarks to the generated tokens incrementally, with maximum detection performance at each token.
The details are deferred to Appendix \ref{App: pf of Lem: token opt watermark} and the algorithm is provided in Section \ref{Sec: algorithm}.

\textbf{Performance Analysis.}  We evaluate the Type-I (FPR) and Type-II ($1-$TPR) errors of this scheme over the entire sequence (cf.\ \eqref{Def: type I and II}). 

\begin{lemma}[(Informal Statement) Token-level optimal watermarking detection errors]\label{Lem: token level errors}
Under the detector $\gamma_{\mathrm{tk}}$ in \eqref{Eq: token level detector} and its corresponding token-level optimal watermarking scheme with  $\eta\!\in\!(0, \min\{ 1, (\alpha/\binom{T}{\lceil T\lambda \rceil})^{\frac{1}{\lceil T\lambda \rceil}}\}]$, for a length-$T$ sequence: 
(i) the worst-case Type-I error  
$\sup_{\scaleto{Q_{X_1^T}}{7pt}}\beta_0\leq \alpha$; (ii) 
if token positions more than $n$ apart are assumed to be independent,   with a suitable detector threshold, the Type-II error decays exponentially in $T/n$.
\end{lemma}
Although the token-level optimal watermarking scheme may not be optimal at the sequence level, we show that it maintains good performance with a proper choice of token-level FPR $\eta$. The formal statement is provided in Appendix \ref{App: pf of Lem: token level errors}. 



Furthermore, we observe that even without explicitly introducing robustness to the token-level optimal watermarking scheme, it inherently leads to some robustness against token replacement. The following result shows that if the auxiliary sequence $\zeta_1^T$ is shared between the LLM and the detector $\gamma_{\mathrm{tk}}$ (cf.\ \eqref{Eq: token level detector}), the token at position $t$ can be replaced with probability $\Pr(\zeta_t \text{ is redundant})$ without affecting detector output.
\begin{proposition}[Robustness against token replacement]\label{Prop: robust against replace}
    Under the detector $\gamma_{\mathrm{tk}}$ in \ \eqref{Eq: token level detector} and its corresponding token-level optimal watermarking scheme, the expected number of tokens that can be randomly replaced in $X_1^T$ without compromising detection performance is $\sum_{t=1}^T\bbE_{X_1^{t-1}}\big[\sum_{x\in\calV}\big(P^*_{X_t|X_1^{t-1}}(x|X_1^{t-1})-\eta\big)_+ \big]$,
with $P^*_{X_t|X_1^{t-1}}$ given in \eqref{Eq: token opt P_X}. 
\end{proposition}

\section{DAWA: \texorpdfstring{\underline{D}}istribution-\texorpdfstring{\underline{A}}daptive \texorpdfstring{\underline{W}}atermarking \texorpdfstring{\underline{A}}lgorithm}\label{Sec: algorithm}
In this section, we implement the token-level design presented in Section \ref{section: token-level} by introducing a novel, distortion-free watermarking algorithm, DAWA (\textbf{D}istribution-\textbf{A}daptive \textbf{W}atermarking \textbf{A}lgorithm). To address the challenge {\bf \circled{3}} of recovering the auxiliary sequence at the detector without knowledge of the original LLM and prompt, we utilize some novel tricks, including a surrogate model and Gumbel-Max sampling, which also ensures model-agnosticism and computational efficiency. 

\paragraph{Novel Trick for Auxiliary Sequence Transmission.}
Since the resulting optimal distribution of the auxiliary variable $\zeta_t$ from Section \ref{section: token-level} is adaptive to the original NTP distribution of  LLM, it is not likely to completely reconstruct it at the detection phase without access to the LLM or prompt.  
One possible workaround is enforcing $P_{\zeta_t} = \Unif(\calZ)$ for both watermark generation and detection. While this method \cite{he2025statistical} simplifies the transmission, it leads to a much higher minimum Type-II error compared to $\beta_1^*(Q_{X_1^T},\alpha, \epsilon)$ (cf.\  \eqref{Eq: Type-II LB}),  indicating a trade-off between detection performance and a non-distribution-adaptive design.

We thus introduce a novel trick to transmit the auxiliary sequence by integrating a \emph{surrogate language model (SLM)} during the detection phase and the Gumbel-Max trick \cite{gumbel1954statistical} for sampling $\zeta_t$. This SLM, much smaller than the watermarked LLM and possibly from a different family (as long as it shares the same tokenizer) approximates the watermarked distributions $\{P_{\scaleto{X_t|X_1^{t-1}}{7pt}}^*\}_{t=1,2,\ldots}$ using only the text \(X_1^T\), without the prompt. With the approximated  $P_{\scaleto{X_t|X_1^{t-1}}{7pt}}^*$, we reconstruct the sampling distribution of $\zeta_t$ and sample it using the Gumbel-Max trick with the $\mathtt{key}$ shared from watermark generation. 

Theoretically, the SLM’s approximation error has limited impact on detection performance, since the watermarking algorithm is provably resilient to token replacement attacks (cf.~Proposition \ref{Prop: robust against replace}).
In Section \ref{Sec: experiments}, our experiments highlight that, even with \emph{incomplete recovery} of $\zeta_1^T$ during detection, the DAWA algorithm with this novel trick exhibits superior detection performance and greater resilience against token replacement attack, surpassing baseline watermarking schemes.


\begin{figure*}[t]
    \centering
    \includegraphics[width=1\linewidth]{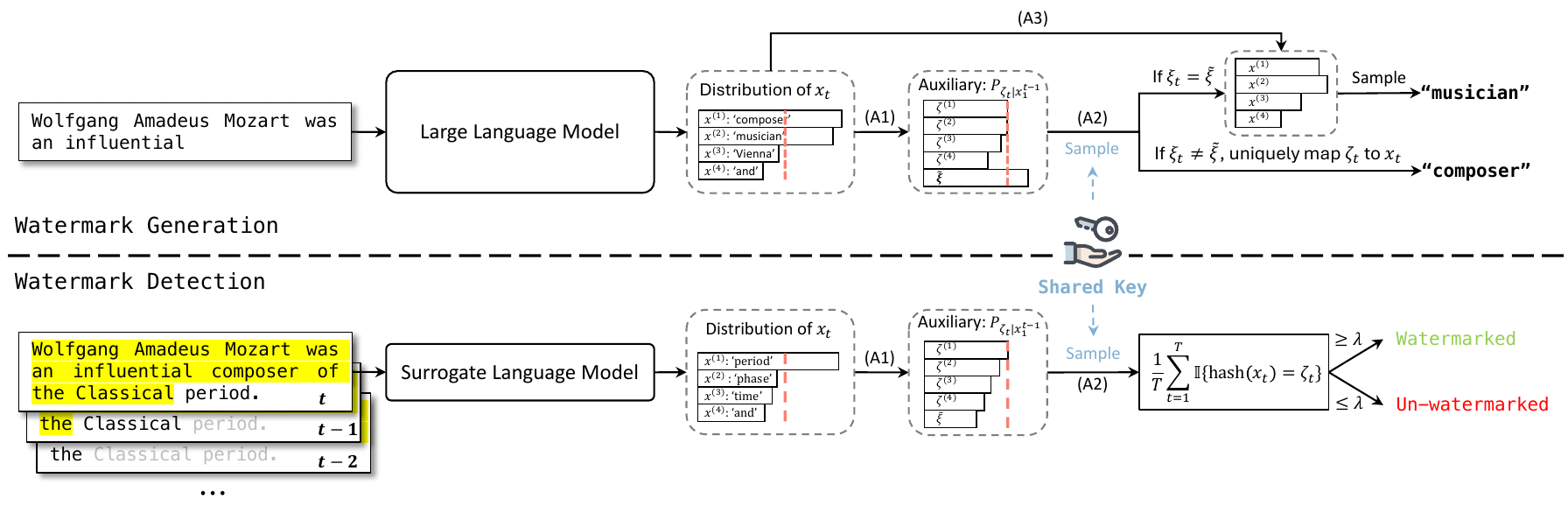}
    \caption{Workflow of our practical algorithm (DAWA) for watermark generation and detection. \eqref{eq:A1}: construct the sampling distribution of auxiliary variable $\zeta_t$ based on $Q_{x_t|x_1^{t-1},\pt}$;  \eqref{eq:A2}: sample $\zeta_t$ using the Gumbel-Max trick and a shared key; \eqref{eq:A3}: adjust the NTP distribution of $x_t$ with $\eta$.  
    }
    \label{fig:watermark_practice}
\end{figure*} 

\paragraph{DAWA Details.}
From the detector in \eqref{Eq: token level detector}, 
we first choose $g_\mathrm{tk}$ that depends on a hash function $h_{\key}$:
    \setlength{\abovedisplayskip}{-1pt}
    \setlength{\abovedisplayshortskip}{-1pt}
\vspace{1pt}
\begin{equation}
    \gamma_{\mathrm{dawa}}(X_1^T,\zeta_1^T)\!=\!\mathbbm{1} \Big\{\frac{1}{T}\sum_{t\in[T]} \! \mathbbm{1}\{h_{\key}(X_t)=\zeta_t\} \!\! \geq \!\! \lambda  \Big\}. \label{Eq: implement dec}
\end{equation}
\vspace{2pt}

DAWA is an implementation of the distortion-free ($\epsilon=0$) token-level optimal watermarking scheme (cf.~Section \ref{section: token-level}), in which the auxiliary variable is sampled adaptively based on the original NTP distribution $Q_{X_t|x_1^{t-1},\pt}$, as illustrated in Figure \ref{fig:watermark_practice} and detailed in Appendix \ref{App: pseudo codes}. Below, we elaborate on the key steps.

\underline{\it Watermark Generation.} 
Using the detector $\gamma_{\mathrm{dawa}}$ from \eqref{Eq: implement dec}, we define the auxiliary alphabet $\calZ$ from unique mappings \(\{h_{\key}(x)\}_{x \in \calV}\) and add a redundant \(\tilde{\zeta}\). At each $t$, $P_{\zeta_t|x_1^{t-1},\pt}$ is adaptive to $Q_{X_t|x_1^{t-1},\pt}$:
    \setlength{\abovedisplayskip}{0pt}
    \setlength{\belowdisplayskip}{0pt}
    \setlength{\abovedisplayshortskip}{0pt}
    \setlength{\belowdisplayshortskip}{0pt}
\begin{small}
\begin{equation}
\left\{\!\!\!
\begin{array}{l}
  P_{\zeta_t|x_1^{t-1},\pt}(\zeta)\leftarrow (Q_{X_t|x_1^{t-1},\pt}(h^{-1}_{\key}(\zeta))\wedge \eta), \quad \forall {\zeta\in\calZ\backslash\{\tilde{\zeta}\}}.  \\
  P_{\zeta_t|x_1^{t-1},\pt}(\tilde{\zeta})\leftarrow \sum_{x\in\calV}(Q_{X_t|x_1^{t-1},\pt}(x)-\eta)_+.
\end{array}
\right.
\tag{A1} \label{eq:A1}
\end{equation}
\end{small}

The Gumbel-Max trick is then used to sample \(\zeta_t\):
\begin{small}
\begin{equation}
    \zeta_t \leftarrow \argmax _{\zeta\in\calZ}\log(P_{\zeta_t|x_1^{t-1},\pt}(\zeta))+G_{t,\zeta}.
     \tag{A2} \label{eq:A2} 
\end{equation}
\end{small}
where $G_{t,\zeta}$ is sampled from the Gumbel distribution using a shared $\mathtt{key}$ and the previous tokens.
If \(\zeta_t\) is non-redundant, let \(x_t = h^{-1}_{\key}(\zeta_t)\); otherwise, $x_t$ is sampled via a multinomial distribution:
\begin{small}
\begin{equation}
    x_t \sim \bigg(\frac{(Q_{X_t|x_1^{t-1},\pt}(x)-\eta)_+}{\sum_{x\in\calV}\big(Q_{X_t|x_1^{t-1},\pt}(x)-\eta\big)_{+}} \bigg)_{x\in \calV}.
    \tag{A3} \label{eq:A3}
\end{equation}
\end{small}

\underline{\it Watermark Detection.} A surrogate NTP distribution \(\tilQ_{X_t|x_1^{t-1}}\) is approximated by the SLM for each \(t\). We then use \eqref{eq:A1} to approximate \(P_{\zeta_t|x_1^{t-1},\pt}\) from \(\tilQ_{X_t|x_1^{t-1}}\) and sample \(\zeta_t\) using \eqref{eq:A2} with the shared $\mathtt{key}$. At each position \(t\), the score \(\mathbbm{1}\{h_{\key}(x_t) = \zeta_t\}\) is $1$ if $\zeta_t$ non-redundant and $0$ otherwise. Compute \(\frac{1}{T}\sum_{t=1}^T\mathbbm{1}\{h_{\key}(x_t) = \zeta_t\}\) and compare with a threshold \(\lambda\). If above \(\lambda\), the text is detected as watermarked.

\section{Experiments and Discussions}\label{Sec: experiments}
\paragraph{Experiment Settings.}
We now introduce the setup details of our experiments.

\underline{\it Implementation Details.} Our approach is implemented on two language models: Llama2-13B \citep{touvron2023llama}, and Mistral-8$\times$7B \citep{jiang2023mistral}. Llama2-7B serves as the surrogate model for Llama2-13B, while Mistral-7B is used as the surrogate model for Mistral-8$\times$7B. We conduct our experiments on Nvidia A100 GPUs. In DAWA, we set $\eta=0.2$ and $T=200$. 


\underline{\it Baselines.} We compare our methods with three existing watermarking methods: KGW+23 \citep{kirchenbauer2023watermark}, EXP-edit \citep{kuditipudi2023robust}, Gumbel-Max \citep{gumbel2023}, and HCW+23 \citep{hu2023unbiased}, where the EXP-edit, Gumbel-Max and HCW+23 are distortion-free watermarks.  KGW+23 employs the prior 1 token as a hash to create a green/red list, with the watermark strength set at 2. 

\begin{table*}[!t]\scriptsize
\centering
\setlength{\tabcolsep}{1pt}
\caption{\small Detection performance on clean and edited text across different LLMs and datasets.}
\label{tab:combined_detection}
\resizebox{1\textwidth}{!}{
\begin{tabular}{llcccccccccccc}
\toprule
\multicolumn{2}{c}{} & \multicolumn{6}{c}{\textbf{Clean Text}} & \multicolumn{6}{c}{\textbf{Token Replacement Attack}} \\
\cmidrule(lr){3-8} \cmidrule(lr){9-14}
\multicolumn{1}{c}{\textbf{LLM}} & \textbf{Method} & 
\multicolumn{3}{c}{C4} & \multicolumn{3}{c}{ELI5} & 
\multicolumn{3}{c}{C4} & \multicolumn{3}{c}{ELI5} \\
\cmidrule(lr){3-5} \cmidrule(lr){6-8} \cmidrule(lr){9-11} \cmidrule(lr){12-14}
& & ROC-AUC & TP@1\% FP & TP@10\% FP & ROC-AUC & TP@1\% FP & TP@10\% FP & ROC-AUC & TP@1\% FP & TP@10\% FP & ROC-AUC & TP@1\% FP & TP@10\% FP \\
\midrule
\multirow{4}{*}{Llama2-13B} 
& KGW+23        & 0.995 & 0.991 & 1.000 & 0.989 & 0.974 & 0.986 & 0.965 & 0.833 & 0.952 & 0.973 & 0.892 & 0.973 \\
& EXP-edit     & 0.986 & 0.968 & 0.996 & 0.983 & 0.960 & 0.995 & 0.973 & 0.857 & 0.978 & 0.967 & 0.889 & 0.975 \\
& Gumbel-Max   & 0.996 & 0.993 & 0.994 & 0.999 & 0.991 & 0.994 & 0.968 & 0.858 & 0.970 & 0.965 & 0.887 & 0.975 \\
& HCW+23 &0.994 &0.928 &0.986 &0.991 &0.888 &0.978 &0.890 &0.268 &0.714 &0.893 &0.278 &0.692 \\
& \cellcolor{gray!30}\textbf{Ours} 
               & \cellcolor{gray!30}0.999 & \cellcolor{gray!30}0.998 & \cellcolor{gray!30}1.000 
               & \cellcolor{gray!30}0.998 & \cellcolor{gray!30}0.997 & \cellcolor{gray!30}1.000 
               & \cellcolor{gray!30}0.989 & \cellcolor{gray!30}0.860 & \cellcolor{gray!30}0.976 
               & \cellcolor{gray!30}0.995 & \cellcolor{gray!30}0.969 & \cellcolor{gray!30}0.994 \\
\midrule
\multirow{4}{*}{Mistral-8$\times$7B} 
& KGW+23        & 0.997 & 0.995 & 1.000 & 0.993 & 0.983 & 0.994 & 0.977 & 0.860 & 0.962 & 0.969 & 0.890 & 0.970 \\
& EXP-edit     & 0.993 & 0.970 & 0.997 & 0.994 & 0.972 & 0.996 & 0.980 & 0.861 & 0.975 & 0.983 & 0.932 & 0.988 \\
& Gumbel-Max   & 0.994 & 0.989 & 0.999 & 0.987 & 0.970 & 0.990 & 0.972 & 0.865 & 0.960 & 0.971 & 0.889 & 0.975 \\
& HCW+23 &0.998 &0.986 &0.994 &0.999 &0.992 &1.000 &0.885 &0.364 &0.674 &0.878 &0.252 &0.668 \\
& \cellcolor{gray!30}\textbf{Ours} 
               & \cellcolor{gray!30}0.999 & \cellcolor{gray!30}0.998 & \cellcolor{gray!30}1.000 
               & \cellcolor{gray!30}0.999 & \cellcolor{gray!30}0.999 & \cellcolor{gray!30}1.000 
               & \cellcolor{gray!30}0.990 & \cellcolor{gray!30}0.881 & \cellcolor{gray!30}0.966 
               & \cellcolor{gray!30}0.993 & \cellcolor{gray!30}0.991 & \cellcolor{gray!30}0.995 \\
\bottomrule
\end{tabular}
}
\end{table*}

\begin{wraptable}[7]{r}{.45\linewidth} 
\vspace{-10pt}
\setlength{\tabcolsep}{10pt}
\scriptsize 
\centering
\caption{Empirical entropy comparison between C4 and ELI5 datasets.}
\label{tab:entropy_comparison}
\begin{tabular}{lcc}
\toprule
\textbf{Model} & C4 (entropy) & ELI5 (entropy) \\
\midrule
Llama2-13B & 0.547 & 0.272 \\
Mistral-8$\times$7B & 1.475 & 1.427 \\
\bottomrule
\end{tabular}
\end{wraptable}
\underline{\it Dataset and Prompt.} Our experiments are conducted using two distinct datasets. The first is an open-ended \textbf{high-entropy} generation dataset, a realnewslike subset from C4 \citep{raffel2020exploring}. The second is a relatively \textbf{low-entropy} generation dataset, ELI5 \citep{fan2019eli5}. The realnewslike subset of C4 is tailored specifically to include high-quality journalistic content that mimics the style and format of real-world news articles. As shown in Table \ref{tab:entropy_comparison}, the C4 dataset consistently exhibits higher empirical entropy than the ELI5 dataset across different models. We utilize the first two sentences of each text as prompts and the following 200 tokens as human-generated text. The ELI5 dataset is specifically designed for the task of long-form question answering (QA), with the goal of providing detailed explanations for complex questions. We use each question as a prompt and its answer as human-generated text.

\underline{\it Evaluation Metrics.} 
To evaluate the performance of watermark detection, we report the ROC-AUC score, where the ROC curve shows the True Positive (TP) Rate against the False Positive (FP) Rate. A higher ROC-AUC score indicates better overall performance. 
The detection threshold $\lambda$ is determined empirically by the ROC-AUC score function based on unwatermarked and watermarked sentences. 

\subsection{Main Experimental Results}
\textbf{Watermark Detection Performance.} 
To explore our detection performance at a very low FPR, we conduct experiments using Llama2-13B on $10^5$ texts ($200$-length) from the Wikipedia dataset and compute the TPR at $1\mathrm{e}{-01}$, $1\mathrm{e}{-02}$, $1\mathrm{e}{-03}$, $1\mathrm{e}{-04}$, and $1\mathrm{e}{-05}$ FPR respectively. Figure \ref{fig:low_fpr} shows that  DAWA significantly outperforms other baselines. 

Furthermore, we compare the detection performance 
across various language models and tasks, as presented in Table \ref{tab:combined_detection}. Our DAWA demonstrates superior performance, especially on the relatively low-entropy QA dataset, validating Theorem \ref{Thm: optimal detector} and Lemma \ref{Lem: token level errors}. This success stems from the design of our watermarking scheme, which reduces the likelihood of low-entropy tokens being falsely detected as watermarked, thereby lowering the FPR. 
 Moreover, this suggests that even without knowing the watermarked LLM during detection, we can still use the proposed SLM and Gumbel-Max trick to successfully detect the watermark.

We assess the robustness of  DAWA against a \textbf{token replacement attack} to validate Proposition \ref{Prop: robust against replace}. 
For each watermarked text, we randomly mask 50\% of the tokens and use T5-large \citep{raffel2020exploring} to predict the replacement for each masked token based on the context. 
Table~\ref{tab:combined_detection} exhibits watermark detection performance under token replacement attacks across different models and tasks. 
Our DAWA remains high ROC-AUC, TPR@1\%FPR, and TPR@10\%FPR under this attack compared with other baselines.  



\begin{table}[] 
\scriptsize
    \centering
    \setlength{\tabcolsep}{10pt}
    \caption{Comparison of BLEU score and perplexity across different watermarking methods.}
    \vspace{1em}
    \begin{tabular}{lcccccc}
        \toprule
        {Methods} & {Human} & {KGW+23} & {EXP-Edit} & {Gumbel-Max} &{HCW+23} &  \textbf{Ours} \\
        \midrule
        BLEU Score $\uparrow$ & 0.219 & 0.158 & 0.203 & 0.210 &0.207 & 0.214 \\
         Perplexity $\downarrow$ & 8.846 & 13.472	&10.126 & 9.910 &10.115	& 10.034  \\
        \bottomrule
    \end{tabular}
    \label{tab:text_quality}
\end{table}

\paragraph{Watermarked Text Quality.} To evaluate the quality of watermarked text generated by our watermarking methods, we report the perplexity (median) on C4 dataset using GPT-3 \citep{brown2020language}, and the BLEU score on the machine translation task using the WMT19 dataset \citep{barrault2019findings} and mBART Model \citep{liu2020mbart}, as shown in Table \ref{tab:text_quality}. 
It can be observed that our scheme achieves a higher BLEU score and a lower perplexity closer to the unwatermarked one (10.020), both close to the score on human datasets. 
This demonstrates that our distortion-free scheme, employing an NTP distribution-adaptive approach, has minimal impact on the generated text quality, preserving its naturalness and coherence.

\paragraph{Ablation Study and Additional Results.} In Appendix \ref{App: add exp}, we further show that (1) our DAWA is efficient and does not affect generation time; (2) detection remains accurate and robust even with a much smaller SLM from a \emph{different model family} and without prompts; (3) TPR increases with longer token length $T$; and (4) our theoretical choice of $\eta$ effectively controls the empirical FPR.

\subsection{Extension  Towards Stronger Robustness}\label{Sec:extension}
In Appendix \ref{app: assess DAWA broad attacks}, Table \ref{tab:more robust detection}, we first empirically assess the robustness of DAWA against \textbf{random deletion and paraphrasing attacks}. DAWA outperforms  Gumbel-Max and KGW+23 in deletion robustness and matches their performance under paraphrasing. These results confirm that DAWA remains competitive while balancing robustness, efficiency, and detection accuracy, with the potential to demonstrate a graceful trade-off based on our theoretical analyses.

\textbf{Theoretical Extension.} As a step towards even stronger robustness,  Appendix~\ref{App: ext to robust} outlines how our theoretical framework and optimal solutions extend to scenarios involving a wide range of attacks, including \emph{semantic-invariant attacks}. We characterize the detectability--distortion--robustness trade-off and show the closed-form optimal \emph{robust} watermarking scheme--detector pairs. These findings offer valuable insights for designing advanced semantic-based watermarking algorithms that are resilient to such attacks in the future.



\section*{Acknowledgements}
This work was performed while Yepeng Liu and Yuheng Bu were with the Department of Electrical and Computer Engineering at the University of Florida. They acknowledge UFIT Research Computing for providing computational resources and support that contributed to the research results reported in this publication. This work was also conducted while Haiyun He was a postdoctoral associate with the Center for Applied Mathematics at Cornell University and while Ziqiao Wang was a Ph.D.~student with the School of Electrical Engineering and Computer Science at the University of Ottawa.

\newpage

\bibliography{sample.bib}
\bibliographystyle{unsrtnat}

\newpage
\appendix

\section{Other Related Literature}\label{App: other literature}
In the past few years, many in-process LLM watermarking methods have been proposed \citep{liu2024survey, wu2025survey, zhou2024bileve, wu2025ensemble, wu2025analyzing, chen2024mark, liu2025context, an2025defending}, including biased and unbiased (distortion-free) ones. Biased watermarks typically alter the next-token prediction (NTP) distribution to increase the likelihood of sampling certain tokens \citep{kirchenbauer2023watermark, zhao2023provable, liu2024adaptive, liu2024a, qu2024provably}. For example, \cite{kirchenbauer2023watermark} divides the vocabulary into green and red lists and slightly enhances the probability of green tokens in the NTP distribution. 
Unbiased watermarks maintain the original NTP distributions or texts unchanged, using various sampling strategies to embed watermarks  \citep{fernandez2023three, fu2024gumbelsoft,wu2023dipmark,zhao2024permute, boroujeny2024multi, christ2024undetectable,giboulot2024watermax, hu2023unbiased, chen2025improved, wu2024distortion}. The Gumbel-Max watermark \citep{gumbel2023} utilizes the Gumbel-Max trick \citep{gumbel1954statistical} to sample the next token, while \citet{kuditipudi2023robust} introduces an inverse transform method for this purpose.

From a theoretical standpoint, most of these designs remain heuristic. While post-process watermarking has been extensively studied from an information-theoretic perspective~\citep{martinian2005authentication, moulin2000information, chen2000design, 1564436, 4418489}, the theory behind in-process watermarking is still limited. Prior efforts \citep{huang2023towards, li2024statistical} have analyzed either the watermark embedder or the detector in isolation, without achieving joint or universally optimal designs.

\section{Other Existing Watermarking Schemes}\label{App: existing watermark}
Here, we discuss additional existing watermarking schemes utilizing auxiliary variables, which can be encompassed within our LLM watermarking formulation.

$\bullet\,$ The \textbf{Gumbel-Max watermarking scheme} \citep{gumbel2023} applies the Gumbel-Max trick \citep{gumbel1954statistical} to sample the next token $X_t$, where the Gumbel variable is exactly the auxiliary variable $\zeta_t$, which is a $|\calV|$-dimensional random vector, indexd by $x$. 
For $t=1,2,\ldots,$
    \begin{enumerate}[noitemsep,topsep=0pt,parsep=0pt,partopsep=0pt, leftmargin=*,label=--]
        \item Compute a hash using the previous $n$ tokens $X_{t-1}^{t-n}$  and a shared secret $\key$, i.e., $h(X_{t-1}^{t-n},\key)$, where  $h:\calV^n\times \bbR\to\bbR$.
        \item Use $h(X_{t-1}^{t-n},\key)$ as a seed to uniformly sample the auxiliary vector $\zeta_t$ from $[0,1]^{|\calV|}$.
        \item Sample $X_t$ using the Gumbel-Max trick
    \begin{equation}
        X_t=\argmax_{x\in\calV} \log Q_{X_t|x_1^{t-1}}(x)-\log(-\log \zeta_t(x)).
    \end{equation}
    \end{enumerate}

    $\bullet\,$ In the \textbf{inverse transform watermarking scheme} \citep{kuditipudi2023robust}, the vocabulary $\calV$ is considered as $[|\calV|]$ and the combination of the uniform random variable and the randomly permuted index vector is the auxiliary variable $\zeta_t$. 
    \begin{enumerate}[noitemsep,topsep=0pt,parsep=0pt,partopsep=0pt, leftmargin=*,label=--]
        \item Use $\key$ as a seed to uniformly and independently sample $\{U_t\}_{t=1}^T$ from $[0,1]$, and $\{\pi_t\}_{t=1}^T$ from the space of permutations over $[|\calV|]$. Let the auxiliary variable $\zeta_t=(U_t,\pi_t)$, for $t=1,2,\ldots,T$.
        \item Sample $X_t$ as follows
    \begin{equation}
        X_t=\pi_t^{-1}\bigg(\min\Big\{i\in[|\calV|]: \sum_{x\in[|\calV|]}\big(Q_{X_t|x_1^{t-1}}(x)\mathbbm{1}\{\pi_t(x)\leq i\}\big) \geq U_t \Big\}\bigg),
    \end{equation}
    where $\pi_t^{-1}$ denotes the inverse permutation.
    \end{enumerate}

   $\bullet\,$ In \textbf{adaptive watermarking} by \citet{liu2024adaptive}, the authors introduce a watermarking scheme that adopts a technique similar to the Green-Red List approach but replaces the hash function with a pretrained neural network $h$. 
   The auxiliary variable $\zeta_t$ is sampled from the set $\{\bv\in\{0,1\}^{|\calV|}: \|\bv\|_1=\rho|\calV|\}$ using the seed $h(\phi(X_{1}^{t-1}),\key)$, where $h$ takes the semantics $\phi(X_{1}^{t-1})$ of the generated text and the secret $\key$ as inputs.  They sample $X_t$ using the same process as the Green-Red List approach.

\section{Proof of Theorem \ref{Thm: universal min typeII}}\label{App: pf of Thm: universal min typeII}

According to the Type-I error constraint, we have $\forall x_1^T\in \calV^T$,
\begin{align}
    \alpha&\geq \max_{Q_{X_1^T}}\bbE_{Q_{X_1^T}P_{\zeta_1^T}}[\mathbbm{1}\{(X_1^T,\zeta_1^T)\in\calA_1\}]\\
    &\geq \bbE_{\delta_{x_1^T}P_{\zeta_1^T}}[\mathbbm{1}\{(X_1^T,\zeta_1^T)\in\calA_1\}]\\
    &=\bbE_{P_{\zeta_1^T}}[\mathbbm{1}\{(x_1^T,\zeta_1^T)\in\calA_1\}]\\
    &=\begin{cases}
        \sum_{\zeta_1^T}P_{\zeta_1^T}(\zeta_1^T)\mathbbm{1}\{(x_1^T,\zeta_1^T)\in\calA_1\}, &\quad \calZ \text{ is discrete;} \\
        \int P_{\zeta_1^T}(\zeta_1^T)\mathbbm{1}\{(x_1^T,\zeta_1^T)\in\calA_1\} \; d \zeta_1^T, &\quad \calZ \text{ is continuous;}
    \end{cases}.
\end{align}

In the following, for notational simplicity, we assume that $\calZ$ is discrete. However, the derivations hold for both discrete $\calZ$ and continuous $\calZ$.
The Type-II error is given by $1-\bbE_{P_{X_1^T,\zeta_1^T}}[\mathbbm{1}\{(X_1^T,\zeta_1^T)\in\calA_1\}]$. We have
\begin{align}\label{Eq: 1-TypeII error}
    \bbE_{P_{X_1^T,\zeta_1^T}}[\mathbbm{1}\{(X_1^T,\zeta_1^T)\in\calA_1\}]=\sum_{x_1^T}\underbrace{\sum_{\zeta_1^T}P_{X_1^T,\zeta_1^T}(x_1^T,\zeta_1^T) \mathbbm{1}\{(x_1^T,\zeta_1^T)\in\calA_1\}}_{C(x_1^T)},
\end{align}
where for all $x_1^T\in\calV^T$, $$C(x_1^T)\leq P_{X_1^T}(x_1^T) \quad \text{and} \quad C(x_1^T)\leq \sum_{\zeta_1^T}P_{\zeta_1^T}(\zeta_1^T) \mathbbm{1}\{(x_1^T,\zeta_1^T)\in\calA_1\}\leq \alpha $$ according to the Type-I error bound. Therefore, 
\begin{align}
    \bbE_{P_{X_1^T,\zeta_1^T}}[\mathbbm{1}\{(X_1^T,\zeta_1^T)\in\calA_1\}]&=\sum_{x_1^T}C(x_1^T) \leq \sum_{x_1^T}(P_{X_1^T}(x_1^T)\wedge \alpha)\\
    &=1-\sum_{x_1^T}(P_{X_1^T}(x_1^T)- \alpha)_+ \label{Eq: 1-TypeI error ub}
\end{align}
where \eqref{Eq: 1-TypeI error ub} is maximized at 
\begin{align}\label{Eq: P_X star}
    P_{X_1^T}^*\coloneqq \argmin_{P_{X_1^T}:\sD(P_{X_1^T}, Q_{X_1^T} )\leq \epsilon}\sum_{x_1^T}(P_{X_1^T}(x_1^T)- \alpha)_+.
\end{align}
For any $P_{X_1^T}$, the Type-II error is lower bounded by
\begin{align}
     \bbE_{P_{X_1^T,\zeta_1^T}}[\mathbbm{1}\{(X_1^T,\zeta_1^T)\notin \calA_1\}]&\geq \sum_{x_1^T}(P_{X_1^T}(x_1^T)- \alpha)_+.
\end{align}
By plugging $P_{X_1^T}^*$ into this lower bound, we obtain a Type-II lower bound that holds for all  $\gamma$ and $P_{X_1^T,\zeta_1^T}$. Recall that \citet{huang2023towards} proposed a type of detector and watermarking scheme that achieved this lower bound. As we demonstrate, it is actually the universal minimum Type-II error over all possible $\gamma$ and $P_{X_1^T,\zeta_1^T}$, denoted by $\beta_1^*(Q_{X_1^T},\epsilon,\alpha)$. 

Specifically, define $\epsilon^*(x_1^T)=Q_{X_1^T}(x_1^T)-P^*_{X_1^T}(x_1^T)$ 
and we have 
\begin{align}
    \sum_{x_1^T:P^*_{X_1^T}(x_1^T)\geq \alpha}\epsilon^*(x_1^T)&= \sum_{x_1^T:P^*_{X_1^T}(x_1^T)\geq \alpha, \epsilon^*(x_1^T) \geq 0}\epsilon^*(x_1^T)+\underbrace{\sum_{x_1^T:P^*_{X_1^T}(x_1^T)\geq \alpha, \epsilon^*(x_1^T) \leq 0}\epsilon^*(x_1^T)}_{\leq 0}\\
    &\leq \sum_{x_1^T:P^*_{X_1^T}(x_1^T)\geq \alpha, \epsilon^*(x_1^T) \geq 0}\epsilon^*(x_1^T)\\
    &=\sum_{x_1^T:P^*_{X_1^T}(x_1^T)\geq \alpha, Q_{X_1^T}(x_1^T)\geq P^*_{X_1^T}(x_1^T) }\epsilon^*(x_1^T)\\
    &\leq \sum_{x_1^T:Q_{X_1^T}(x_1^T)\geq P^*_{X_1^T}(x_1^T)}\epsilon^*(x_1^T)\leq\epsilon
\end{align}
where the last inequality follows from the total variation distance constraint $\sD_\TV(P_{X_1^T}, Q_{X_1^T} )\leq \epsilon$.
We rewrite $\beta_1^*(Q_{X_1^T},\epsilon,\alpha)$ as follows:
\begin{align}
     \beta_1^*(Q_{X_1^T},\epsilon,\alpha)&= \min_{P_{X_1^T}:\sD_{\TV}(P_{X_1^T}, Q_{X_1^T} )\leq \epsilon}\sum_{x_1^T}(P_{X_1^T}(x_1^T)- \alpha)_+ \label{Eq: Type-II LB 1}\\
     &= \sum_{x_1^T: P^*_{X_1^T}(x_1^T)\geq \alpha}(P^*_{X_1^T}(x_1^T)- \alpha),\\
     &=\sum_{x_1^T: P^*_{X_1^T}(x_1^T)\geq \alpha}(Q_{X_1^T}(x_1^T)- \epsilon^*(x_1^T)-\alpha)\\\\
     &=\sum_{x_1^T: P^*_{X_1^T}(x_1^T)\geq \alpha}(Q_{X_1^T}(x_1^T)-\alpha)- \sum_{x_1^T: P^*_{X_1^T}(x_1^T)\geq \alpha}\epsilon^*(x_1^T)\\
     &\geq  \sum_{x_1^T}(Q_{X_1^T}(x_1^T)- \alpha)_+-\epsilon, 
\end{align}
where the  last inequality follows from $\sum_{x_1^T:P^*_{X_1^T}(x_1^T)\geq \alpha} \epsilon^*(x_1^T)\leq \epsilon$, i.e. the total variation constraint limits how much the distribution $P_{X_1^T}^*$ can be perturbed from $Q_{X_1^T}$.
Since $\beta_1^*(Q_{X_1^T},\epsilon,\alpha) \geq 0$, finally we have
\begin{align}
    \beta_1^*(Q_{X_1^T},\epsilon,\alpha)&\geq \bigg(\sum_{x_1^T}(Q_{X_1^T}(x_1^T)- \alpha)_+-\epsilon\bigg)_+.
\end{align}
 Notably, the lower bound is achieved when $\{x_1^T: P^*_{X_1^T}(x_1^T)\geq \alpha\}=\{x_1^T: Q_{X_1^T}(x_1^T)\geq P^*_{X_1^T}(x_1^T)\}$ and $\sD_\TV(Q_{X_1^T},P^*_{X_1^T})=\epsilon$. That is, to construct $P^*_{X_1^T}$, an $\epsilon$ amount of the mass of $Q_{X_1^T}$ above $\alpha$ is moved to below $\alpha$, which is possible only when $\sum_{x_1^T}(\alpha-Q_{X_1^T}(x_1^T))_+\geq\epsilon$.  Note that  \citet[Theorem~3.2]{huang2023towards} points out a sufficient condition for this to hold: $|\calV|^T\geq \frac{1}{\alpha}$.
The optimal distribution $P^*_{X_1^T}$ thus satisfies
\begin{align}
    \sum_{x_1^T: Q_{X_1^T}(x_1^T)\geq \alpha}(Q_{X_1^T}(x_1^T)-P^*_{X_1^T}(x_1^T))=\sum_{x_1^T: Q_{X_1^T}(x_1^T)\leq \alpha}(P^*_{X_1^T}(x_1^T)-Q_{X_1^T}(x_1^T))=\epsilon.
\end{align}

\paragraph{Refined constraints for optimization.} We notice that the feasible region of \eqref{Eq: opt-I} can be further reduced as follows:
\begin{align}
  \min_{P_{X_1^T}}&\min_{ P_{   \zeta_1^T | X_1^T }} \quad \bbE_{ P_{X_1^T} P_{   \zeta_1^T | X_1^T }}[1-\gamma(X_1^T,\zeta_1^T)] \tag{Opt-II} \label{Eq:opt-II}\\
    \text{s.t.} 
   \quad & \int P_{ \zeta_1^T| X_1^T    }(  \zeta_1^T | x_1^T ) \; d \zeta_1^T = 1, \; \forall x_1^T \nn\\
   \quad & \int P_{ \zeta_1^T| X_1^T    }(  \zeta_1^T | x_1^T )\gamma(x_1^T,\zeta_1^T)  \leq 1\wedge \frac{\alpha}{P_{X_1^T}(x_1^T)}, \; \forall x_1^T \label{Eq: cond prob constraint} \\
   \quad & \sD_{\TV}(P_{X_1^T}, Q_{  X_1^T }) \leq \epsilon, \nn\\ 
\quad & \sup_{Q_{X_1^T}} \sum_{x_1^T}Q_{X_1^T}(x_1^T) \int \bigg(\sum_{y_1^T} P_{\zeta_1^T|X_1^T}(\zeta_1^T|y_1^T)P_{X_1^T}(y_1^T) \bigg)\gamma(x_1^T,\zeta_1^T) \; d \zeta_1^T \leq \alpha, \nn
\end{align}

where \eqref{Eq: cond prob constraint} is an additional constraint on $P_{ \zeta_1^T| X_1^T }$.
If and only if \eqref{Eq: cond prob constraint} can be achieved with equality, the minimum of the objective function $\bbE_{ P_{X_1^T} P_{   \zeta_1^T | X_1^T }}[1-\gamma(X_1^T,\zeta_1^T)]$ reaches  \eqref{Eq: Type-II LB}. 

\section{Formal Statement of Theorem \ref{Thm: optimal detector} and its Proof}\label{App: pf of Thm: optimal detector}
\newtheorem*{thm:opt dec}{Theorem~\ref{Thm: optimal detector} [Formal]}
\begin{thm:opt dec}[Optimal type of detectors and watermarking schemes] 
The set of all detectors that achieve the minimum Type-II error $\beta_1^*(Q_{X_1^T},\alpha, \epsilon)$
 in Theorem \ref{Thm: universal min typeII} for all text distribution $Q_{X_1^T}\in\calP(\calV^T)$ and distortion level $\epsilon \ge 0$ is precisely 
 \begin{equation}
    \begin{aligned}
        \Gamma^*\coloneqq \big\{\gamma\mid &\gamma(X_1^T,\zeta_1^T)= \mathbbm{1}\{X_1^T=g(\zeta_1^T)\},
        \text{ for some surjective } g:\calZ^T\to \cS \supset \cV^T \big\}.
    \end{aligned}
\end{equation}
For any valid function $g$, choose a redundant auxiliary value $\tilde{\zeta_1^T}\in\calZ^T$ such that $x_1^T\ne g(\tilde{\zeta_1^T})$ for all $x_1^T\in\calV^T$. The detailed construction of the optimal watermarking scheme is as follows:   
 \begin{equation}  P_{X_1^T}^*=\min_{P_{X_1^T}:\sD(P_{X_1^T}, Q_{X_1^T} )\leq \epsilon}\sum_{x_1^T}(P_{X_1^T}(x_1^T)- \alpha)_+, 
\end{equation}
\begin{align}\label{Eq: opt cond distr}
    &\text{and for any } x_1^T\in \calV^T, P_{\zeta_1^T|X_1^T}^*(\zeta_1^T|x_1^T) \text{ satisfies } \\
    &\begin{cases}
         P_{X_1^T}^*(x_1^T)\sum_{\zeta_1^T}   P_{\zeta_1^T|X_1^T}^*(\zeta_1^T|x_1^T) \gamma(x_1^T,\zeta_1^T) =P_{X_1^T}^*(x_1^T)\wedge \alpha, &\quad \forall \zeta_1^T \text{ s.t. } \gamma(x_1^T,\zeta_1^T)=1;\\
          P_{X_1^T}^*(x_1^T)   P_{\zeta_1^T|X_1^T}^*(\zeta_1^T|x_1^T) =\big(P_{X_1^T}^*(x_1^T)- \alpha \big)_+,\label{Eq: empty set prob} &\quad \text{if }  \zeta_1^T=\tilde{\zeta_1^T};\\
          P_{\zeta_1^T|X_1^T}^*(\zeta_1^T|x_1^T)=0, &\quad \text{otherwise.}
    \end{cases}
\end{align}
\end{thm:opt dec}

\begin{proof}
 First, we observe that the lower bound on the Type-II error in \eqref{Eq: Type-II LB} is attained if and only if the constraint in \eqref{Eq: cond prob constraint} holds with equality for all $x_1^T$ and for the optimizer. Thus, it suffices to show that for any detector $\gamma\notin\Gamma^*$, the constraint in \eqref{Eq: cond prob constraint} cannot hold with equality for all $x_1^T$ given any text distributions $Q_{X_1^T}$. First, define an arbitrary surjective function $g:\calZ^T\to\calS$, where $\calS$ is on the same metric space as $\calV^T$. Cases 1 and 2 prove that $\calV^T\subset \calS$. Case 3 proves that $\gamma$ can only be $\gamma(X_1^T,\zeta_1^T)= \mathbbm{1}\{X_1^T=g(\zeta_1^T)\}$.
\setlist[itemize]{leftmargin=*}
 \begin{itemize}
 \item \underline{\textbf{Case 1:} $\gamma(X_1^T,\zeta_1^T)= \mathbbm{1}\{X_1^T=g(\zeta_1^T)\}$ but $\calS\subset \calV^T$.} There exists $\tilx_1^T$ such that for all $\zeta_1^T$, $\mathbbm{1}\{\tilx_1^T=g(\zeta_1^T)\}=0$. 
 Under this case, \eqref{Eq: cond prob constraint} cannot hold with equality for $\tilx_1^T$ since the LHS is always $0$ while the RHS is positive.
 
 \item \underline{\textbf{Case 2:} $\gamma(X_1^T,\zeta_1^T)= \mathbbm{1}\{X_1^T=g(\zeta_1^T)\}$ but $\calS=\calV^T$.}    Let us start from the simple case where $T=1$, $\calV=\{x_1,x_2\}$, $\calZ=\{\zeta_1,\zeta_2\}$, and $g$ is an identity mapping. Given any $Q_{X}$ and any feasible $P_X$ such that $\sD_\TV(P_X,Q_X)\leq \epsilon$,  when \eqref{Eq: cond prob constraint} holds with equality, i.e.,
    \begin{align}
        P_{X,\zeta}(x_1,\zeta_1)=P_{X}(x_1)\wedge \alpha \quad \text{and}\quad  P_{X,\zeta}(x_2,\zeta_2)=P_{X}(x_2)\wedge \alpha, 
    \end{align}
    then the marginal $P_\zeta$ is given by: $P_\zeta(\zeta_1)=P_{X}(x_1)\wedge \alpha+(P_{X}(x_2)- \alpha)_+$, $P_\zeta(\zeta_2)=P_{X}(x_2)\wedge \alpha+(P_{X}(x_1)- \alpha)_+$. The worst-case Type-I error is given by  
    \begin{align}
        &\sup_{Q_X} \Big(Q_X(x_1)\big(P_{X}(x_1)\wedge \alpha+(P_{X}(x_2)- \alpha)_+ \big)+ Q_X(x_2)\big(P_{X}(x_2)\wedge \alpha+(P_{X}(x_1)- \alpha)_+\big)\Big)  \\
        &\geq P_{X}(x_1)\wedge \alpha+(P_{X}(x_2)- \alpha)_+\\
        &>\alpha, \quad \text{ if } P_{X}(x_1)>\alpha, P_{X}(x_2)>\alpha.
    \end{align}
    It implies that for any $Q_X$ such that $\{P_X\in\calP(\calV): \sD_\TV(P_X,Q_X)\leq \epsilon\}\subseteq \{P_X\in\calP(\calV): P_{X}(x_1)>\alpha, P_{X}(x_2)>\alpha\}$, the false-alarm constraint is violated  when \eqref{Eq: cond prob constraint} holds with equality. It can be verified that this result also holds for larger $(T,\calV,\calZ)$ and other functions $g:\calZ^T\to \calV^T$.


\item \underline{\textbf{Case 3:} Let $\Xi_\gamma(x_1^T)\coloneqq \{\zeta_1^T\in\calZ^T: \gamma(x_1^T,\zeta_1^T)=1\}$. $\exists x_1^T\ne y_1^T \in \calV^T, \text{ s.t. } \Xi(x_1^T)\cap \Xi(y_1^T)\ne\emptyset$.} For any detector $\gamma\notin \Gamma^*$ that does not fall into Cases 1 and 2, it falls into Case 3. 
 Let us start from the simple case where $T=1$, $\calV=\{x_1,x_2\}$, $\calZ=\{\zeta_1,\zeta_2,\zeta_3\}$.  
    Consider a detector $\gamma$ as follows: $\gamma(x_1,\zeta_1)=\gamma(x_2,\zeta_1)=1$ and $\gamma(x,\zeta)=0$ for all other pairs $(x,\zeta)\in\calV\times \calZ$. Hence,  $\Xi(x_1)\cap \Xi(x_2)=\{\zeta_1\}$. When \eqref{Eq: cond prob constraint} holds with equality, i.e.,
    \begin{align}
        P_{X,\zeta}(x_1,\zeta_1)=P_{X}(x_1)\wedge \alpha \quad \text{and}\quad  P_{X,\zeta}(x_2,\zeta_1)=P_{X}(x_2)\wedge \alpha, 
    \end{align}
    we have the worst-case Type-I error lower bounded by
    \begin{align}
        \sup_{Q_X} \bigg(Q_X(x_1)P_\zeta(\zeta_1)+ Q_X(x_2)P_\zeta(\zeta_1)\bigg) &= P_\zeta(\zeta_1) = P_{X}(x_1)\wedge \alpha+ P_{X}(x_2)\wedge \alpha\\
        &>\alpha, \quad \text{ if } P_{X}(x_1)>\alpha \text{ or } P_{X}(x_2)>\alpha.
    \end{align}
     Thus, for any $Q_X$ such that $\{P_X\in\calP(\calV): \sD_\TV(P_X,Q_X)\leq \epsilon\}\subseteq \{P_X\in\calP(\calV): P_{X}(x_1)>\alpha \text{ or } P_{X}(x_2)>\alpha\}$, the false-alarm constraint is violated  when \eqref{Eq: cond prob constraint} holds with equality.

    If we consider a detector $\gamma$ as follows: $\gamma(x_1,\zeta_1)=\gamma(x_2,\zeta_1)=\gamma(x_2,\zeta_2)=1$ and $\gamma(x,\zeta)=0$ for all other pairs $(x,\zeta)\in\calV\times \calZ$. We still have $\Xi(x_1)\cap \Xi(x_2)=\{\zeta_1\}$. 
    When \eqref{Eq: cond prob constraint} holds with equality, i.e.,
    \begin{align}
        P_{X,\zeta}(x_1,\zeta_1)=P_{X}(x_1)\wedge \alpha \quad \text{and}\quad  P_{X,\zeta}(x_2,\zeta_1)+P_{X,\zeta}(x_2,\zeta_2)=P_{X}(x_2)\wedge \alpha, 
    \end{align}
    we have the worst-case Type-I error lower bounded by
    \begin{align}
        &\sup_{Q_X} \bigg(Q_X(x_1)P_\zeta(\zeta_1)+ Q_X(x_2)(P_\zeta(\zeta_1)+P_\zeta(\zeta_2))\bigg) =\sup_{Q_X} \bigg(P_\zeta(\zeta_1)+ Q_X(x_2)P_\zeta(\zeta_2)\bigg)\\
        &=P_\zeta(\zeta_1)+P_\zeta(\zeta_2)=P_{X}(x_1)\wedge \alpha+P_{X}(x_2)\wedge \alpha>\alpha, \quad \text{ if } P_{X}(x_1)>\alpha \text{ or } P_{X}(x_2)>\alpha,
    \end{align}
    which is the same as the previous result.

    If we let $\calV=\{x_1,x_2,x_3\}$, $\calZ=\{\zeta_1,\zeta_2,\zeta_3,\zeta_4\}$ and $\gamma(x_3,\zeta_3)=1$ in addition to the aforementioned $\gamma$, we can similarly show that the worst-case Type-I error is larger than $\alpha$ for some distributions $Q_X$. 
    
    Therefore, it can be observed that as long as $ \Xi(x_1^T)\cap \Xi(y_1^T)\ne\emptyset$ for some $x_1^T\ne y_1^T \in \calV^T$, \eqref{Eq: cond prob constraint} can not be achieved with equality for all $Q_{X_1^T}$ and $\epsilon$ even for larger $(T,\calV,\calZ)$ as well as continuous $\calZ$. 
 \end{itemize}
In conclusion, for any detector $\gamma\notin \Gamma^*$, the universal minimum Type-II error in \eqref{Eq: Type-II LB} cannot be obtained for all $Q_{X_1^T}$ and $\epsilon$.

Since the optimal detector takes the form  $\gamma(X_1^T,\zeta_1^T)= \mathbbm{1}\{X_1^T=g(\zeta_1^T)\}$ for some surjective function $g:\calZ^T \to \calS$, $\calS\supset\calV^T$, and the token vocabulary is discrete, it suffices to consider discrete $\calZ$ to derive the optimal watermarking scheme.

Under the watermarking scheme $P_{X_1^T,\zeta_1^T}^*$ (cf.~\eqref{Eq: P_X star} and \eqref{Eq: opt cond distr}), the Type-I and Type-II errors are given by:

\textbf{Type-I error:}
\begin{align}
     & \forall y_1^T\in\calV^T, \quad  \bbE_{P_{\zeta_1^T}^*}[\mathbbm{1}\{y_1^T=g(\zeta_1^T)\}]=\sum_{\zeta_1^T}P_{\zeta_1^T}^*(\zeta_1^T)\mathbbm{1}\{y_1^T=g(\zeta_1^T)\}\\
    &=\sum_{\zeta_1^T}\sum_{x_1^T}P_{X_1^T,\zeta_1^T}^*(x_1^T,\zeta_1^T)\mathbbm{1}\{y_1^T=g(\zeta_1^T)\}\\
    &=P_{X_1^T}^*(y_1^T)\sum_{\zeta_1^T}   P_{\zeta_1^T|X_1^T}^*(\zeta_1^T|y_1^T)\mathbbm{1}\{y_1^T=g(\zeta_1^T)\} =P_{X_1^T}^*(y_1^T)\wedge \alpha \\
    &\leq \alpha,\\
    & \!\!\!\! \text{and since any distribution $Q_{X_1^T}$ can be written as a linear combinations of $\delta_{y_1^T}$, we have} \nn\\
&\max_{Q_{X_1^T}}\bbE_{Q_{X_1^T}P_{\zeta_1^T}^*}[\mathbbm{1}\{X_1^T=g(\zeta_1^T)\}]\leq \alpha.
\end{align}
\textbf{Type-II error:}
\begin{align}
     \quad &1-\bbE_{P^*_{X_1^T,\zeta_1^T}}[\mathbbm{1}\{X_1^T=g(\zeta_1^T)\}] \nn\\
    &=1-\sum_{x_1^T}\sum_{\zeta_1^T}P_{X_1^T,\zeta_1^T}^*(x_1^T,\zeta_1^T) \mathbbm{1}\{x_1^T=g(\zeta_1^T)\}\\
    &=1-\sum_{x_1^T}P_{X_1^T}^*(x_1^T)\sum_{\zeta_1^T}   P_{\zeta_1^T|X_1^T}^*(\zeta_1^T|x_1^T)\mathbbm{1}\{x_1^T=g(\zeta_1^T)\} \\
    &=1-\sum_{x_1^T}\big(P_{X_1^T}^*(x_1^T)\wedge \alpha \big)\\
    &=\sum_{x_1^T:P_{X_1^T}^*(x_1^T)>\alpha}(P_{X_1^T}^*(x_1^T)- \alpha).
\end{align}
The optimality of $P_{X_1^T,\zeta_1^T}^*$ is thus proved. We note that \eqref{Eq: cond prob constraint} in \eqref{Eq:opt-II} holds with equality under this optimal conditional distribution $P_{\zeta_1^T|X_1^T}^*$. 

Compared to \citet[Theorem 3.2]{huang2023towards}, their proposed detector is equivalent to $\gamma(X_1^T,\zeta_1^T)= \mathbbm{1}\{X_1^T=\zeta_1^T\}$, where $\calZ^T=\calV^T\cup\{\tilde{\zeta}_1^T\}$ and $\tilde{\zeta}_1^T\notin \calV^T$, meaning that it belongs to $\Gamma^*$. 
\end{proof}

\section{Illustration of  Construction of the Optimal Watermarking Scheme}\label{App: toy example of opt watermark}
Using a toy example in Figure~\ref{Fig: watermarking}, we now illustrate how to construct the optimal watermarking schemes, where 
\begin{equation}
P_{X_1^T}^*=\argmin_{P_{X_1^T}:\sD(P_{X_1^T}, Q_{X_1^T} )\leq \epsilon}\sum_{x_1^T}(P_{X_1^T}(x_1^T)- \alpha)_+. 
\end{equation}

\begin{figure}[ht]
    \centering
    \includegraphics[trim={0 .3cm .6cm .4cm},  width=.55\linewidth]{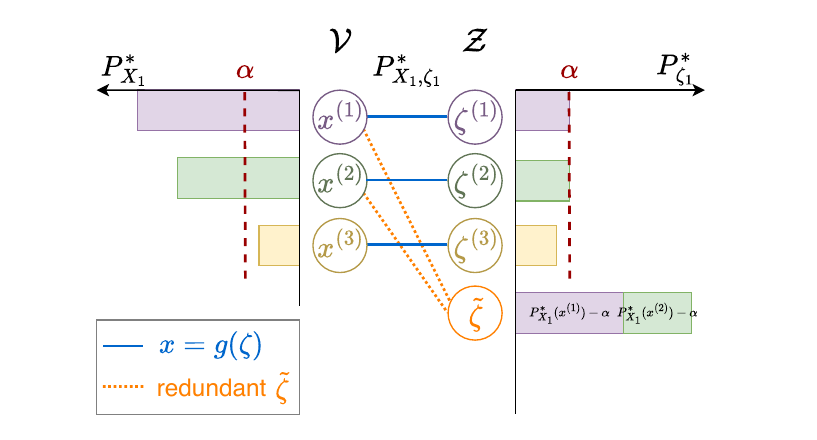} \caption{A toy example of the optimal detector and watermarking scheme when $T=1$. Links between $\calV$ and $\calZ$ suggest $P_{X_1,\zeta_1}^*> 0$. }
    \label{Fig: watermarking}
\end{figure}
Constructing the optimal watermarking scheme $P_{X_1^T,\zeta_1^T}^*$ is equivalent to transporting the probability mass $P_{X_1^T}^*$ on $\calV$ to $\calZ$, maximizing $P_{X_1^T,\zeta_1^T}^*(x_1^T,\zeta_1^T)$ when $x_1^T=g(\zeta_1^T)$, while keeping the worst-case Type-I error below $\alpha$.  Without loss of generality, by letting $T=1$, we present Figure \ref{Fig: watermarking} to visualize the optimal watermarking scheme. The construction process is given step by step as follows:

\vspace{-.5em}
\textbf{-- Identify text-auxiliary pairs:} We begin by identifying text-auxiliary pairs $(x,\zeta)\in\calV\times \calZ$ with $\gamma(x,\zeta)=\mathbbm{1}\{x=g(\zeta)\}=1$ and connect them by blue solid lines.

\vspace{-.5em}
\textbf{-- Introducing redundant auxiliary value:} We enlarge $\calZ$ to include an additional value $\tilde{\zeta}$ and set $\gamma(x, \tilde{\zeta})=0$ for all $x$. We will call $\tilde{\zeta}$ ``redundant".

\vspace{-.5em}
\textbf{-- Mass allocation for $P_{X_1}^*(x)>\alpha$:} If $P_{X_1}^*(x)>\alpha$, we transfer $\alpha$ mass of $P_{X_1}^*(x)$ to the $\zeta$ connected by the blue solid lines. The excess mass is transferred to the redundant $\tilde{\zeta}$ (orange dashed lines).
Specifically, 
for $x^{(1)}$, where $P_{X_1}^*(x^{(1)})>\alpha$ and $x^{(1)}=g(\zeta^{(1)})$, we move $\alpha$ units of mass from  $P_{X_1}^*(x^{(1)})$ to $P_{\zeta_1}^*(\zeta^{(1)})$, ensuring that $P_{\zeta_1}^*(\zeta^{(1)})=\alpha$. The rest $(P_{X_1}^*(x^{(1)})-\alpha)$ units of mass is moved to $\tilde{\zeta}$. Similarly, for $x^{(2)}$, where $P_{X_1}^*(x^{(2)})>\alpha$ and $x^{(2)}=g(\zeta^{(2)})$, we move $\alpha$ mass from $P_{X_1}^*(x^{(2)})$ to $P_{\zeta_1}^*(\zeta^{(2)})$ and $(P_{X_1}^*(x^{(2)})-\alpha)$ mass to $\tilde{\zeta}$. Consequently, the probability of $\tilde{\zeta}$ is $P_{\zeta_1}(\tilde{\zeta})=(P_{X_1}^*(x^{(1)})-\alpha)+(P_{X_1}^*(x^{(2)})-\alpha)$. In this way, there is a chance for the lower-entropy texts $x^{(1)}$ and $x^{(2)}$ to be mapped to the redundant $\tilde{\zeta}$ during watermark generation.

\vspace{-.5em}
\textbf{-- Mass allocation for $P_{X_1}^*(x)<\alpha$:}
For $x^{(3)}$, where $P_{X_1}^*(x^{(3)})<\alpha$ and $x^{(3)}=g(\zeta^{(3)})$, we move the entire mass $P_{X_1}^*(x^{(3)})$ to $P_{\zeta_1}^*(\zeta^{(3)})$ along the blue solid line. It means that higher-entropy texts will not be mapped to the redundant $\tilde{\zeta}$ during watermark generation.

\vspace{-.5em}
\textbf{-- Outcome:}
    This construction ensures that $P_{\zeta_1}^*(\zeta)\leq \alpha$ for all $\zeta\in\{\zeta^{(1)},\zeta^{(2)},\zeta^{(3)}$, keeping the worst-case Type-I error under control. The Type-II error is equal to $P_{\zeta_1}^*(\tilde{\zeta})$, which is exactly the universally minimum Type-II error. This scheme can be similarly generalized to $T>1$.

In Figure \ref{Fig: watermarking}, when there is no link between $(x,\zeta)\in\calV\times \calZ$, the joint probability $P_{X_1,\zeta_1}^*(x,\zeta)=0$.
By letting $\epsilon=0$, the scheme guarantees that the watermarked LLM remains unbiased (distortion-free). Note that the detector proposed in \citet[Theorem 3.2]{huang2023towards} is also included in our framework, see Appendix~\ref{App: pf of Thm: optimal detector}.


\section{Construction of Token-level Optimal Watermarking Scheme}\label{App: pf of Lem: token opt watermark}
The toke-level optimal watermarking scheme is the optimal solution to the following optimization problem:
\begin{align}
    &\inf_{P_{X_t,\zeta_t|X_1^{t-1},\zeta_1^{t-1}}} \bbE_{P_{X_t,\zeta_t|X_1^{t-1},\zeta_1^{t-1}}}[1-\mathbbm{1}\{X_t=g_\mathrm{tk}(\zeta_t)\}] \nn\\
    & \qquad \quad \mathrm{s.t.}   \sup_{Q_{X_t|X_1^{t-1}}}\bbE_{Q_{X_t|X_1^{t-1}}\otimes P_{\zeta_t|\zeta_1^{t-1}}}[\mathbbm{1}\{X_t=g_\mathrm{tk}(\zeta_t)\}]\leq \eta, \; \sD_\TV(P_{X_t|X_1^{t-1}}, Q_{X_t|X_1^{t-1}})\leq \epsilon.  \nn
\end{align}
The optimal solution $P_{X_t,\zeta_t|X_1^{t-1},\zeta_1^{t-1}}^*$ follows the similar rule as that of $P_{X_1^T,\zeta_1^T}^*$ in Theorem \ref{Thm: optimal detector} with $(Q_{X_1^T},P_{X_1^T},\alpha)$ replaced by $(Q_{X_t|X_1^{t-1}},P_{X_t|X_1^{t-1}},\eta)$. We refer readers to Appendix \ref{App: pf of Thm: optimal detector} for further details.

\section{Formal Statement of Lemma \ref{Lem: token level errors} and its Proof}\label{App: pf of Lem: token level errors}
Let $P_{X_1^T,\zeta_1^T}^{\mathrm{token}*}$ and $P_{\zeta_1^T}^{\mathrm{token}*}$ denote the joint distributions induced by the token-level optimal watermarking scheme.
\newtheorem*{lem:formal}{Lemma \ref{Lem: token level errors} (Formal)}
\begin{lem:formal}[Token-level optimal watermarking detection errors]
    Let $\eta=(\alpha/\binom{T}{\lceil T\lambda \rceil})^{\frac{1}{\lceil T\lambda \rceil}}$. Under the detector $\gamma$ in \eqref{Eq: token level detector} and the token-level optimal  watermarking scheme $P_{X_t,\zeta_t|X_1^{t-1},\zeta_1^{t-1}}^*$,  the Type-I error is upper bounded by 
    \begin{equation}
        \sup_{Q_{X_1^T}}\beta_0(\gamma,Q_{X_1^T},P_{\zeta_1^T}^{\mathrm{token}*})\leq \alpha.
    \end{equation}
    Assume that when $T$ and $n\leq T$ are both large enough,  token $X_t$ is independent of  $X_{t-i}$, i.e., $P_{X_t,X_{t-i}}=P_{X_t}\otimes P_{X_{t-i}}$, for all $i\geq n+1$ and $t\in[T]$. 
     Let $\calI_{T,n}(i)=([i-n, i+n]\cap [T])\backslash \{i\}$.
    By setting the detector threshold as $\lambda=\frac{a}{T}\sum_{t=1}^T\bbE_{X_t,\zeta_t}[\mathbbm{1}\{X_t=g(\zeta_t)\}]$ for some $a\in[0,1]$, the Type-II error exponent is
    \[-\log \beta_1(\gamma,P_{X_1^T,\zeta_1^T}^{\mathrm{token}*})=\Omega\bigg(\frac{T}{n} \bigg).\]

\end{lem:formal}

The following is the proof of Lemma \ref{Lem: token level errors}.

To choose $\lceil T\lambda \rceil$ indices out of $\{1,\ldots,T\}$, there are $\binom{T}{\lceil T\lambda \rceil}$ choices. Let $k=1,\ldots, \binom{T}{\lceil T\lambda \rceil}$ and $S_k$ be the $k$-th set of the chosen indices. The Type-I error is upper-bounded by
\begin{align}
    \beta_0(\gamma,Q_{X^{(T)}},P_{\zeta_1^T}^{\mathrm{token}*})&=\Pr\bigg(\frac{1}{T}\sum_{t=1}^T\mathbbm{1}\{X_t=g(\zeta_t)\} \geq \lambda \mid \rmH_0 \bigg)\\
    &\leq \Pr\bigg(\bigcup_{k=1}^{\binom{T}{\lceil T\lambda \rceil} }\{\mathbbm{1}\{X_t=g(\zeta_t)\}=1,\forall t\in S_k \} \mid \rmH_0 \bigg)\\
    &\leq \sum_{k=1}^{\binom{T}{\lceil T\lambda \rceil} } \underbrace{\Pr\bigg(\{\mathbbm{1}\{X_t=g(\zeta_t)\}=1,\forall t\in S_k \} \mid \rmH_0 \bigg)}_{P_{\text{FA},k}}.
\end{align}
Without loss of generality, let $m=\lceil T\lambda\rceil $ and $S_k=\{1,2,\ldots, m \}$. We can rewrite $P_{\text{FA},k}$ as
\begin{align}
    P_{\text{FA},k}&=\bbE_{Q_{X^{(T)}}\otimes P_{\zeta^{(T)}}}[\{\mathbbm{1}\{X_t=g(\zeta_t)\}=1,\forall t\in S_k \}]\\
    &=\bbE_{Q_{X^{(T)}}\otimes P_{\zeta^{(T)}}}[\prod_{t\in S_k}\mathbbm{1}\{X_t=g(\zeta_t)\}]\\
    &=\bbE_{Q_{X_1}\otimes P_{\zeta_1}}\Bigg[\mathbbm{1}\{X_1=g(\zeta_1)\} \bbE_{Q_{X_2|X_1}\otimes P_{\zeta_2|\zeta_1}}\bigg[\mathbbm{1}\{X_2=g(\zeta_2)\}\cdots \\
    &\qquad\qquad\qquad\qquad\qquad\qquad \cdots \bbE_{Q_{X_m|X_1^{m-1}}\otimes P_{\zeta_m|\zeta_1^{m-1}}}[\mathbbm{1}\{X_m=g(\zeta_m)\}]\big]\cdots \bigg] \Bigg]\\
    &\leq \eta^m, \quad \forall Q_{X_1^T}.
\end{align}
Then the Type-I error is finally upper-bounded by
\begin{align}
    \sup_{Q_{X_1^T}}\beta_0(\gamma,Q_{X_1^T},P_{\zeta_1^T}^{\mathrm{token}*})\leq \binom{T}{\lceil T\lambda \rceil}\eta^{\lceil T\lambda \rceil}\leq\alpha.
\end{align}

We prove the Type-II error bound by applying
\citet[Theorem 10]{janson1998new}.

\begin{theorem}[Theorem 10, {\citet{janson1998new}}]
    Let $\{I_i\}_{i\in\calI}$ be a finite family of indicator random variables, defined on a common probability space. Let $G$ be a dependency graph of $\calI$, i.e., a graph with vertex set $\calI$ such that if $A$ and $B$ are disjoint subsets of  $\calI$, and $\Gamma$ contains no edge between $A$ and $B$, then  $\{I_i\}_{i\in A}$ and $\{I_i\}_{i\in B}$ are independent. We write $i\sim j$ if $i,j\in\calI$ and $(i,j)$ is an edge in $G$. In particular, $i\not\sim i$. Let $S=\sum_{i\in\calI}I_i$ and $\Delta=\bbE[S]$. Let $\Psi=\max_{i\in\calI} \sum_{j\in\calI, j\sim i}\bbE[I_j]$ and $\Phi=\frac{1}{2}\sum_{i\in\calI}\sum_{j\in\calI,j\sim i}\bbE[I_i I_j]$. For any $0\leq a \leq 1$,
    \begin{align}
        \Pr(S\leq a\Delta) \leq \exp\bigg\{-\min\bigg\{(1-a)^2\frac{\Delta^2}{8\Phi+2\Delta},(1-a)\frac{\Delta}{6\Psi} \bigg\}\bigg\}.\label{Eq: tail inequality}
    \end{align}
\end{theorem}

Given any detector $\gamma$ that accepts the form in \eqref{Eq: token level detector} and the corresponding optimal watermarking scheme, for some $a\in(0,1)$,
we first set the threshold in $\gamma$ as
\begin{align}
    T\lambda=a\sum_{t=1}^T\bbE_{X_t,\zeta_t}[\mathbbm{1}\{X_t=g(\zeta_t)\}]=a\sum_{t=1}^T\bbE_{X_1^{t-1}}\bigg[\sum_{x}\big(P^*_{X_t|X_1^{t-1}}(x|X_1^{t-1})-\eta \big)_+ \bigg]=: a \Delta_T,
\end{align}
where $P^*_{X_t|X_1^{t-1}}$ is induced by $P^*_{X_t,\zeta_t|X_1^{t-1},\zeta_1^{t-1}}$.
The Type-II error is given by
\begin{align}
    \beta_1(\gamma,P_{X_1^T,\zeta_1^T}^{\mathrm{token}*})=P_{X_1^T,\zeta_1^T}^{\mathrm{token}*}\bigg(\sum_{t=1}^T\mathbbm{1}\{X_t=g(\zeta_t)\} < a \Delta_T \bigg)
\end{align}
which is exactly the left-hand side of \eqref{Eq: tail inequality}. 

Assume that when $T$ and $n\leq T$ are large enough,  token $X_t$ is independent of all $X_{t-i}$ for all $i\geq n+1$ and $t\in[T]$, i.e., $P_{X_t,X_{t-i}}=P_{X_t}\otimes P_{X_{t-i}}$. Let $\calI_{T,n}(i)=([i-n, i+n]\cap [T])\backslash \{i\}$.
The $\Psi$ and $\Phi$ on the right-hand side of \eqref{Eq: tail inequality} are given by:
\begin{align}
    \Psi&\coloneqq\max_{i\in[T]}\sum_{t\in[T],t\sim i}\bbE_{X_t,\zeta_t}[\mathbbm{1}\{X_t=g(\zeta_t)\}]=\max_{i\in[T]}\sum_{t\in\calI_{T,n}(i)}\bbE_{X_t,\zeta_t}[\mathbbm{1}\{X_t=g(\zeta_t)\}]=\Theta(n),
\end{align}
\begin{align}
    \Phi&\coloneqq \frac{1}{2}\sum_{i\in[T]}\sum_{j\in[T],j\sim i}\bbE[\mathbbm{1}\{X_i=g(\zeta_i)\} \mathbbm{1}\{X_j=g(\zeta_j)\}]\\
    &=\frac{1}{2}\sum_{i\in[T]}\sum_{j\in\calI_{T,n}(i)}\bbE[\mathbbm{1}\{X_i=g(\zeta_i)\} \mathbbm{1}\{X_j=g(\zeta_j)\}]=\Theta(Tn).
\end{align}
By plugging $\Delta_T$, $\Omega$ and $\Theta$ back into the right-hand side of \eqref{Eq: tail inequality}, we have the upper bound  
\begin{align}
        \beta_1(\gamma,P_{X_1^T,\zeta_1^T}^{\mathrm{token}*})\leq \exp\bigg\{-\min\bigg\{(1-a)^2\frac{\Delta_T^2}{8\Phi+2\Delta_T},(1-a)\frac{\Delta_T}{6\Psi} \bigg\}\bigg\} \label{Eq: token type-ii error ub}
    \end{align}
    where $U_t=\bbE_{X_1^{t-1}}\big[\sum_{x}\big(P^*_{X_t|X_1^{t-1}}(x|X_1^{t-1})-\eta \big)_+ \big]$, $\Delta_T\coloneqq \sum_{t=1}^TU_t$, $\Psi=\max_{i\in[T]}\sum_{t\in\calI_{T,n}(i)}U_t$, and $\Phi=\frac{1}{2}\sum_{i\in[T]}\sum_{j\in\calI_{T,n}(i)}\bbE[\mathbbm{1}\{X_i=g(\zeta_i)\} \mathbbm{1}\{X_j=g(\zeta_j)\}]$.
This implies
\begin{align}
    &-\log \beta_1(\gamma,P_{X_1^T,\zeta_1^T}^{\mathrm{token}*})\geq \min\bigg\{(1-a)^2 \Theta\bigg(\frac{T}{n}\bigg),(1-a)\Theta\bigg(\frac{T}{n}\bigg) \bigg\}\\
    \Longrightarrow &-\log \beta_1(\gamma,P_{X_1^T,\zeta_1^T}^{\mathrm{token}*})=\Omega\bigg(\frac{T}{n}\bigg).
\end{align}

\section{DAWA  Pseudo-Codes}\label{App: pseudo codes}

\scalebox{0.9}{
\begin{minipage}{1.05\linewidth}
\begin{algorithm}[H]
    \caption{Watermarked Text Generation}
    \label{algorithm1: generation}
    \begin{algorithmic}[1]
        \REQUIRE LLM $Q$, Vocabulary $\mathcal{V}$, Prompt $u$, Secret $\key$, Token-level false alarm $\eta$.
        \STATE $\calZ=\{h_{\key}(x)\}_{x\in\calV}\cup \{\tilde{\zeta}\}$ 
        \vspace{2pt}
        \FOR {$t=1,...,T$ } 
            \vspace{1pt}
            \STATE  $P_{\zeta_t|x_1^{t-1},u}(\zeta)\leftarrow (Q_{X_t|x_1^{t-1},u}(h^{-1}_{\key}(\zeta))\wedge \eta), ~\forall {\zeta\in\calZ\backslash\{\tilde{\zeta}\}}$.
            \STATE $P_{\zeta_t|x_1^{t-1},u}(\tilde{\zeta})\leftarrow \sum_{x\in\calV}(Q_{X_t|x_1^{t-1},u}(x)-\eta)_+$.
            \STATE Compute a hash of tokens $x_{t-n}^{t-1}$ with $\key$, and use it as a seed to  generate $(G_{t,\zeta})_{\zeta\in\calZ}$ from Gumbel distribution.
            \STATE $\zeta_t \leftarrow \argmax _{\zeta\in\calZ}\log(P_{\zeta_t|x_1^{t-1},u}(\zeta))+G_{t,\zeta}$.
            \IF{$\zeta_t \neq \tilde{\zeta}$}
                \STATE $x_t \leftarrow h_{\key}^{-1}(\zeta_t)$
            \ELSE
                \STATE Sample $x_t \sim \bigg(\frac{(Q_{X_t|x_1^{t-1},u}(x)-\eta)_+}{\sum_{x\in\calV}\big(Q_{X_t|x_1^{t-1},u}(x)-\eta\big)_{+}} \bigg)_{x\in \calV}$
                \vspace{-2pt}
            \ENDIF
            \vspace{1pt}
        \ENDFOR
        \ENSURE Watermarked text $x_1^T=(x_1, ..., x_T)$.
        \vspace{-0.7pt}
    \end{algorithmic}
\end{algorithm}
\end{minipage}
}

\scalebox{0.9}{
\begin{minipage}{1.05\linewidth}
\begin{algorithm}[H]
    \caption{Watermarked Text Detection}
    \label{algorithm1: detection}
    \begin{algorithmic}[1]
        \REQUIRE SLM $\tilQ$, Vocabulary $\mathcal{V}$, Text $x_1^T$, Secret $\key$, Token-level false alarm $\eta$, Threshold $\lambda$.
        \STATE score $=0$, \; $\calZ=\{h_{\key}(x)\}_{x\in\calV}\cup \{\tilde{\zeta}\}$
        \vspace{3pt}
        \FOR{$t=1,...,T$}
            \vspace{1pt}
            \STATE  $\tilP_{\zeta_t|x_1^{t-1}}(\zeta)\leftarrow (\tilQ_{X_t|x_1^{t-1}}(h^{-1}_{\key}(\zeta))\wedge \eta), ~\forall {\zeta\in\calZ\backslash\{\tilde{\zeta}\}}$.
            \STATE $\tilP_{\zeta_t|x_1^{t-1}}(\tilde{\zeta})\leftarrow \sum_{x\in\calV}(\tilQ_{X_t|x_1^{t-1}}(x)-\eta)_+$.
            \vspace{2pt}
            \STATE Compute a hash of tokens $x_{t-n}^{t-1}$ with $\key$, and use it as a seed to generate $(G_{t,\zeta})_{\zeta\in\calZ}$ from Gumbel distribution.
            \vspace{2pt}
            \STATE $\zeta_t \leftarrow \argmax _{\zeta\in\calZ}\log(\tilP_{\zeta_t|x_1^{t-1}}(\zeta))+G_{t,\zeta}$.
            \STATE score $\leftarrow$ score $+  \mathbbm{1}\{h_{\key}(x_t) =\zeta_t\}$
            \vspace{3pt}
        \ENDFOR
        \vspace{3pt}
        \IF{score $>T\lambda$}
            \RETURN 1
            \textcolor{gray}{\COMMENT{Input text is watermarked}}
        \ELSE
            \RETURN 0
            \textcolor{gray}{\COMMENT{Input text is unwatermarked}}
        \ENDIF
        \vspace{0.12cm}
    \end{algorithmic}
\end{algorithm}
\end{minipage}
}

\section{Ablation Study and Additional Experimental Results}\label{App: add exp}

All pre-existing models and datasets utilized in this research are publicly available and were used in full accordance with their respective licensing terms, which predominantly include common open-source licenses (such as Apache 2.0, MIT, CC BY-SA) and specific community or research usage agreements.

\paragraph{Efficiency of Watermark Scheme.}  To evaluate the efficiency of our watermarking method, we conduct experiments to measure the average generation time for both watermarked and unwatermarked text. In both scenarios, we generated $500$ texts, each containing $200$ tokens. Table \ref{tab:gen_time_comparison} indicates that the difference in generation time between unwatermarked and watermarked text is less than $0.5$ seconds. This minimal difference confirms that our watermarking method has a negligible impact on generation speed, ensuring practical applicability.

\begin{table}[h]\scriptsize
\centering
\caption{Average generation time comparison for watermarked and unwatermarked text using Llama2-13B.}
\begin{tabular}{lcc}
\toprule
\textbf{Language Model} & \textbf{Setting}      & \textbf{Avg Generation Time (s)} \\ \midrule
Llama2-13B              & Unwatermarked          & 9.110                            \\
Llama2-13B              & Watermarked            & 9.386                            \\ \bottomrule
\end{tabular}
\label{tab:gen_time_comparison}
\end{table}

\begin{table}[h]\scriptsize
\centering
\caption{Performance comparison of different language models and surrogate models under two scenarios: without attack and with token replacement attack.}
\vspace{1em}
\begin{tabular}{llcccc}
\toprule
\textbf{Scenario}         & \textbf{Language Model} & \textbf{Surrogate Model} & \textbf{ROC-AUC} & \textbf{TPR@1\% FPR} & \textbf{TPR@10\% FPR} \\ \midrule
\multirow{3}{*}{Without Attack} 
                          & Llama2-13B             & Llama2-7B               & $0.999$          & $0.998$               & $1.000$               \\ 
                          & Mistral-8 × 7B         & Mistral-7B             & $0.999$          & $0.998$               & $1.000$               \\ 
                          & GPT-J-6B               & GPT-2 large            & $0.997$          & $0.990$               & $0.997$               \\ \midrule
\multirow{3}{*}{With Attack} 
                          & Llama2-13B             & Llama2-7B              & $0.989$          & $0.860$               & $0.976$               \\ 
                          & Mistral-8 × 7B         & Mistral-7B             & $0.990$          & $0.881$               & $0.966$               \\ 
                          & GPT-J-6B               & GPT-2 large            & $0.987$          & $0.892$               & $0.962$               \\ \bottomrule
\end{tabular}
\label{tab:combined_attack_table}
\end{table}
\paragraph{Surrogate Language Model.} SLM plays a crucial role during the detection process of our watermarking method. We examine how the choice of SLM affects the detection performance of our watermarking scheme. The selection of the surrogate model is primarily based on its vocabulary or tokenizer rather than the specific language model within the same family. This choice is critical because, during detection, the text must be tokenized exactly using the same tokenizer as the watermarking model to ensure accurate token recovery. As a result, any language model that employs the same tokenizer can function effectively as the surrogate model. To validate our approach, we apply our watermarking algorithm to GPT-J-6B (a model with 6 billion parameters) and use GPT-2 Large (774 million parameters) as the SLM. Despite differences in developers, training data, architecture, and training methods, these two models share the same tokenizer, making them compatible for this task. We conduct experiments using the C4 dataset, and the results are presented in Table \ref{tab:combined_attack_table}. The results demonstrate the effectiveness of our proposed watermarking method with or without attack, even when using a surrogate model from a different family than the watermarking language model. Notably, the surrogate model, despite having fewer parameters and lower overall capability compared to the watermarking language model, does not compromise the watermarking performance.

\paragraph{Prompt Agnostic.} Prompt agnosticism is a crucial property of LLM watermark detection. We investigate the impact of prompts on our watermark detection performance by conducting experiments to compare detection accuracy with and without prompts attached to the watermarked text during the detection process. The results are presented in Table \ref{tab:prompt_attack_comparison}. Notably, even when prompts are absent and the SLM cannot perfectly reconstruct the same distribution of $\zeta_t$ as in the generation process, our detection performance remains almost unaffected. This demonstrates the robustness of our watermarking method, regardless of whether a prompt is included during the detection phase.

\begin{table*}[h]\scriptsize
\centering
\caption{Performance comparison of Llama2-13B under two scenarios: without attack and with token replacement attack, with and without prompts.}
\begin{tabular}{llccccc}
\toprule
\textbf{Scenario}         & \textbf{Language Model} & \textbf{Surrogate Model} & \textbf{Setting}        & \textbf{ROC-AUC} & \textbf{TPR@1\% FPR} & \textbf{TPR@10\% FPR} \\ \midrule
\multirow{2}{*}{Without Attack} 
                          & Llama2-13B             & Llama2-7B               & Without Prompt          & $0.997$          & $0.983$               & $0.995$               \\ 
                          & Llama2-13B             & Llama2-7B               & With Prompt             & $0.998$          & $0.989$               & $0.996$               \\ \midrule
\multirow{2}{*}{With Attack} 
                          & Llama2-13B             & Llama2-7B               & Without Prompt          & $0.977$          & $0.818$               & $0.953$               \\ 
                          & Llama2-13B             & Llama2-7B               & With Prompt             & $0.979$          & $0.816$               & $0.960$               \\ \bottomrule
\end{tabular}
\label{tab:prompt_attack_comparison}
\end{table*}

\paragraph{Detection Performance with larger $T$.} Increasing text length generally improves detection performance for LLM watermarking. We conduct an additional experiment with $T=500$, and the results are shown below. Both DAWA and KGW+23 show improved performance compared to $T=200$ (reported in Figure \ref{fig:low_fpr}). Notably, DAWA, a distortion-free algorithm, achieves significantly better detection in the low-FPR regime than the distorted KGW+23.
\begin{table}[H]\scriptsize
\centering
\caption{Detection performance at various FPRs for different sequence lengths $T$.}
\vspace{1em}
\label{tab:tpr_vs_fpr}
\begin{tabular}{cccccc}
\toprule
\textbf{Length $T$} & \textbf{Method} & \textbf{TPR@1e-5FPR} & \textbf{TPR@1e-4FPR} & \textbf{TPR@1e-3FPR} & \textbf{TPR@1e-2FPR} \\
\midrule
\multirow{2}{*}{500} & KGW+23 & 0.876 & 0.959 & 0.986 & 0.995 \\
                     & Ours  & \textbf{0.891} & \textbf{0.970} & \textbf{0.996} & \textbf{0.999} \\
\midrule
\multirow{2}{*}{200} & KGW+23 & 0.682 & 0.916 & 0.976 & 0.991 \\
                     & Ours  & \textbf{0.882} & \textbf{0.951} & \textbf{0.992} & \textbf{0.997} \\
\bottomrule
\end{tabular}
\end{table}

\paragraph{Empirical analysis on False Alarm Control.}
We conduct experiments to show the relationship between theoretical FPR (i.e., $\alpha$) and the corresponding empirical FPR. 
As discussed in Lemma \ref{Lem: token level errors}, we set the token-level false alarm rate as $\eta=0.1$ and the sequence length as $T=200$, which controls the sequence-level false alarm rate under $\alpha = \binom{T}{\lceil T \lambda\rceil} \eta^{\lceil T \lambda\rceil}$, where $\lambda$ is the detection threshold. 
For a given theoretical FPR $\alpha$, we calculate the corresponding threshold $\lambda$ and the empirical FPR based on $100$k unwatermarked sentences.
The results, as shown in Table~\ref{exp:false_alarm_control}, confirm that our theoretical guarantee effectively controls the empirical false alarm rate.

\begin{table}[h]\scriptsize
\centering
\setlength{\tabcolsep}{8pt}
\captionsetup{font={small}}
\caption{Theoretical and empirical FPR under different thresholds.}
\label{exp:false_alarm_control}
\vspace{1em}
\begin{tabular}{lllll}
\toprule
Theoretical FPR &9e-03  &2e-03  &5e-04  &9e-05   \\ \midrule
Empirical FPR   &1e-04  &4e-05  &2e-05  &2e-05   \\ \bottomrule
\end{tabular}
\end{table}

\section{Theoretical Extension to Robustness against Broader Attacks}\label{App: ext to robust}
\vspace{-0.5em}
Thus far, we have theoretically examined the optimal detector and watermarking scheme without considering adversarial scenarios. In practice, users may attempt to modify LLM output to remove watermarks through techniques like replacement, deletion, insertion, paraphrasing, or translation. We now show that our framework can be extended to incorporate robustness against these attacks. 

\subsection{Assessment of DAWA against Deletion and Paraphrasing Attacks}\label{app: assess DAWA broad attacks}

We conducted additional experiments on random deletion attacks and paraphrasing attacks, where DAWA achieves comparable robustness to Gumbel-Max and KGW+23. Although it is less robust than EXP-edit, we note that EXP-edit explicitly includes robustness designs and is significantly slower in detection. This demonstrates that DAWA remains competitive while balancing robustness, efficiency, and detection accuracy, with the potential to demonstrate a graceful tradeoff based on our theoretical analyses.
\begin{table}[H]\scriptsize
\centering
\caption{Detection performance under deletion and paraphrasing attacks.}
\vspace{1em}
\label{tab:more robust detection}
\begin{tabular}{lcccccc}
\toprule
\multirow{2}{*}{\textbf{Method}} & \multicolumn{3}{c}{\textbf{Deletion Attacks}} & \multicolumn{3}{c}{\textbf{Paraphrasing Attacks}} \\
\cmidrule(lr){2-4} \cmidrule(lr){5-7}
& ROC-AUC & TPR@1\%FPR & TPR@10\%FPR & ROC-AUC & TPR@1\%FPR & TPR@10\%FPR \\
\midrule
KGW+23         & 0.895 & 0.523 & 0.809 & 0.769 & 0.156 & 0.455 \\
Gumbel-Max    & 0.910 & 0.501 & 0.823 & 0.773 & 0.152 & 0.463 \\
EXP-edit      & 0.978 & 0.955 & {0.970} & {0.853} & {0.245} & {0.703} \\
DAWA (Ours)   & 0.918 & 0.504 & 0.812 & 0.770 & 0.144 & 0.458 \\
\bottomrule
\end{tabular}
\end{table}

\subsection{Theoretical Analysis of \texorpdfstring{$f$}~-Robust Design}
We consider a broad class of attacks, where the text can be altered in arbitrary ways as long as certain latent pattern, such as its \emph{semantics}, is preserved. 
Specifically, let $f:\calV^T \to [K]$ be a function that maps a sequence of tokens $X_1^T$ to a finite latent space $[K]\subset \bbN_+$; for example, $[K]$ may index $K$ distinct semantics clusters and $f$ is a function extracting the semantics. Clearly, $f$ induces an equivalence relation, say, denoted by $\equiv_f$, on $\cV^T$, where $x_1^T \equiv_f {x'}_1^T$ if and only if $f(x_1^T) = f({x'}_1^T)$. Let $\calB_f(x_1^T)$ be an equivalence class containing $x_1^T$.
Under the assumption that the adversary is arbitrarily powerful except that it is unable to move any $x_1^T$ outside its equivalent class $\calB_f(x_1^T)$ (e.g., unable to alter the semantics of $x_1^T$), the 
``$f$-robust'' Type-I and Type-II errors are then defined as\looseness=-3
\begin{align}  
\beta_0(\gamma,Q_{X_1^T},P_{\zeta_1^T},f)&\coloneqq\bbE_{Q_{X_1^T}\otimes P_{\zeta_1^T}}\left[\sup\nolimits_{\substack{\tilx_1^T \in \calB_f(X_1^T)}}\mathbbm{1}\{\gamma(\tilx_1^T,\zeta_1^T )=1\} \right], \\ \beta_1(\gamma,P_{X_1^T,\zeta_1^T},f)&\coloneqq\bbE_{P_{X_1^T,\zeta_1^T}}\left[\sup\nolimits_{\tilx_1^T\in \calB_f(X_1^T)}\mathbbm{1}\{\gamma(\tilx_1^T,\zeta_1^T )=0\} \right].
\end{align}
 Designing a universally optimal $f$-robust detector and watermarking scheme can then be formulated as
 jointly minimizing the $f$-robust Type-II error while constraining the worst-case $f$-robust Type-I error, namely, solving the optimization problem 
\begin{align}
    \inf_{\gamma,P_{X_1^T,\zeta_1^T}}  &\beta_1(\gamma,P_{X_1^T,\zeta_1^T},f)\quad 
    \text{ s.t. }  \sup_{Q_{X_1^T}}\beta_0(\gamma,Q_{X_1^T},P_{\zeta_1^T},f)\leq \alpha, \; \sD_\TV(P_{X_1^T},Q_{X_1^T})\leq \epsilon. \tag{Opt-R} \label{Eq: opt-R}
\end{align}
We prove the following theorem.
\begin{theorem}[Universally minimum $f$-robust Type-II error]\label{Thm: robust min error}
    The universally minimum $f$-robust Type-II error attained from \eqref{Eq: opt-R} is
    \begin{small}
    \begin{equation}
        \beta_1^*(Q_{X_1^T},\alpha,\epsilon,f)\coloneqq \min\limits_{P_{X_1^T}:\sD(P_{X_1^T},Q_{X_1^T})\leq \epsilon}\sum_{k\in[K]}\bigg(\bigg(\sum_{x_1^T:f(x_1^T)=k}P_{X_1^T}(x_1^T) \bigg)- \alpha \bigg)_+. 
    \end{equation}
    \end{small}
\end{theorem}
\vspace{-5pt}

Notably, $\beta_1^*(Q_{X_1^T},\alpha,\epsilon,f)$ 
is suboptimal without an adversary but becomes optimal under the adversarial setting of (Opt-R). The gap between $\beta_1^*(Q_{X_1^T},\alpha,\epsilon,f)$ in Theorem \ref{Thm: robust min error} and $\beta_1^*(Q_{X_1^T},\alpha,\epsilon)$ in Theorem \ref{Thm: universal min typeII} reflects the cost of ensuring robustness, widening as $K$ decreases (i.e., as perturbation strength increases), see Figure~\ref{Fig: uni type ii robust} in appendix for an illustration of the optimal $f$-robust minimum Type-II error when $f$ is a semantic mapping. Similar to Theorem \ref{Thm: optimal detector}, we derive the optimal detector and watermarking scheme achieving $\beta_1^*(Q_{X_1^T},\alpha,\epsilon,f)$, detailed in Appendix \ref{App: optimal robust detector}. These solutions closely resemble those in Theorem \ref{Thm: optimal detector}. For implementation, if the latent space $[K]$ is significantly smaller than $\mathcal{V}^T$, applying the optimal $f$-robust detector and watermarking scheme becomes more effective than those presented in Theorem \ref{Thm: optimal detector}. Additionally, a similar algorithmic strategy to the one discussed in Sections \ref{section: token-level} and \ref{Sec: algorithm} can be employed to address the practical challenges discussed earlier. These extensions and efficient implementations of the function $f$ in practice are promising directions of future research.

\section{Proof of Theorem \ref{Thm: robust min error}}\label{App: pf of Thm: robust min error}

According to the Type-I error constraint, we have $\forall x_1^T\in \calV^T$,
\begin{align}
    \alpha&\geq \max_{Q_{X_1^T}}\bbE_{Q_{X_1^T}\otimes P_{\zeta_1^T}}\left[\sup_{\tilx_1^T\in \calB_f(X_1^T)}\mathbbm{1}\{\gamma(\tilx_1^T,\zeta_1^T )=1\} \right]\\
    &\geq \bbE_{\delta_{x_1^T}\otimes P_{\zeta_1^T}}\left[\sup_{\tilx_1^T\in \calB_f(X_1^T)}\mathbbm{1}\{\gamma(\tilx_1^T,\zeta_1^T )=1\} \right]=\bbE_{P_{\zeta_1^T}}\left[\sup_{\tilx_1^T\in \calB_f(x_1^T)}\gamma(\tilx_1^T,\zeta_1^T ) \right]\\
    &=\sum_{\zeta_1^T}P_{\zeta_1^T}(\zeta_1^T)\sup_{\tilx_1^T\in \calB_f(x_1^T)}\gamma(\tilx_1^T,\zeta_1^T ).
\end{align}

For brevity, let $\calB(k)\coloneqq\calB_f(x_1^T)$ if $f(x_1^T)=k$.
The $f$-robust Type-II error is equal to $1-\bbE_{P_{X_1^T,\zeta_1^T}}[\inf_{\tilx_1^T\in \calB_f(X_1^T)}\gamma(\tilx_1^T,\zeta_1^T ) ]$. We have
\begin{align}\label{Eq: 1-robustTypeII error}
    &\bbE_{P_{X_1^T,\zeta_1^T}}\left[\inf_{\tilx_1^T\in \calB_f(X_1^T)}\gamma(\tilx_1^T,\zeta_1^T) \right]\leq \bbE_{P_{X_1^T,\zeta_1^T}}\left[\sup_{\tilx_1^T\in \calB_f(X_1^T)}\gamma(\tilx_1^T,\zeta_1^T) \right]\\
    &=\sum_{k\in[K]}\underbrace{\sum_{x_1^T:f(x_1^T)=k}\sum_{\zeta_1^T}P_{X_1^T,\zeta_1^T}(x_1^T,\zeta_1^T) \sup_{\tilx_1^T\in \calB_f(x_1^T)}\gamma(\tilx_1^T,\zeta_1^T)}_{C(k)},
\end{align}
where according to the $f$-robust Type-I error constraint, for all $k\in[K]$, 
\begin{equation*}
    C(k)\leq \sum_{x_1^T:f(x_1^T)=k}P_{X_1^T}(x_1^T), \quad \text{and}
\end{equation*}
\begin{align*}
    C(k)&=\sum_{\zeta_1^T}P_{\zeta_1^T}(\zeta_1^T)\sum_{x_1^T:f(x_1^T)=k}P_{X_1^T|\zeta_1^T}(x_1^T|\zeta_1^T)\sup_{\tilx_1^T\in \calB(k)}\gamma(\tilx_1^T,\zeta_1^T) \\
    &\leq \sum_{\zeta_1^T}P_{\zeta_1^T}(\zeta_1^T) \sup_{\tilx_1^T\in \calB(k)}\gamma(\tilx_1^T,\zeta_1^T)\leq \alpha.
\end{align*}
Therefore, 
\begin{align}
    &\bbE_{P_{X_1^T,\zeta_1^T}}\left[\inf_{\tilx_1^T\in \calB(f(X_1^T))}\gamma(\tilx_1^T,\zeta_1^T) \right]\leq \sum_{k\in[K]}C(k) \\
    &\leq \sum_{k\in[K]}\bigg(\bigg(\sum_{x_1^T:f(x_1^T)=k}P_{X_1^T}(x_1^T)\bigg) \wedge \alpha \bigg)=1-\sum_{k\in[K]}\bigg(\bigg(\sum_{x_1^T:f(x_1^T)=k}P_{X_1^T}(x_1^T)\bigg)- \alpha \bigg)_+, \label{Eq: 1-robustTypeI error ub}
\end{align}
where \eqref{Eq: 1-robustTypeI error ub} is maximized by taking
\begin{align}
    P_{X_1^T}=P_{X_1^T}^{*,f}\coloneqq \argmin_{P_{X_1^T}: \sD(P_{X_1^T},Q_{X_1^T})\leq \epsilon} \sum_{k\in[K]}\bigg(\bigg(\sum_{x_1^T:f(x_1^T)=k}P_{X_1^T}(x_1^T)\bigg)- \alpha \bigg)_+.
\end{align}
For any $P_{X_1^T}$, the $f$-robust Type-II error is lower bounded by
\begin{align}
     \bbE_{P_{X_1^T,\zeta_1^T}}\left[\sup_{\tilx_1^T\in \calB_f(X_1^T)}\mathbbm{1}\{\gamma(\tilx_1^T,\zeta_1^T )=0\} \right]&\geq \sum_{k\in[K]}\bigg(\bigg(\sum_{x_1^T:f(x_1^T)=k}P_{X_1^T}(x_1^T)\bigg)- \alpha \bigg)_+.
\end{align}
By plugging $P_{X_1^T}^{*,f}$ into the lower bound, we obtain the universal minimum $f$-robust Type-II error over all possible $\gamma$ and $P_{X_1^T,\zeta_1^T}$, denoted by 
\begin{align}
     \beta_1^*(f, Q_{X_1^T},\epsilon,\alpha)&\coloneqq \min_{P_{X_1^T}: \sD(P_{X_1^T},Q_{X_1^T})\leq \epsilon} \sum_{k\in[K]}\bigg(\bigg(\sum_{x_1^T:f(x_1^T)=k}P_{X_1^T}(x_1^T)\bigg)- \alpha \bigg)_+. \label{Eq: robustType-II LB}
\end{align}

\section{Optimal Type of \texorpdfstring{$f$}~-Robust Detectors and Watermarking Schemes}\label{App: optimal robust detector}

\begin{theorem}[Optimal type of $f$-robust detectors and watermarking schemes] \label{Thm: optimal robust detector}
    Let $\Gamma^*_f$ be a collection of detectors that accept the form 
\begin{align}
    \gamma(X_1^T,\zeta_1^T)= \mathbbm{1}\{X_1^T=g(\zeta_1^T) \text{ or } f(X_1^T)=g(\zeta_1^T)\}
\end{align}
for some function $g:\calZ^T\to \calS$, $\calS\cap ([K]\cup \calV^T)\ne \emptyset$ and $|\calS|> K$. 
If and only if the detector $\gamma\in\Gamma^*_f$,
the minimum Type-II error attained from  \eqref{Eq: opt-R} reaches $\beta_1^*( Q_{X_1^T},\epsilon,\alpha,f)$ in \eqref{Eq: robustType-II LB} for all text distribution $Q_{X_1^T}\in\calP(\calV^T)$ and distortion level $\epsilon\in \bbR_{\geq 0}$.

After enlarging $\calZ^T$ to include redundant auxiliary values,  the $\epsilon$-distorted optimal $f$-robust watermarking scheme  $P_{X_1^T,\zeta_1^T}^{*,f}(x_1^T,\zeta_1^T)$ is given as follows: 
    \begin{align}
        P_{X_1^T}^{*,f}\coloneqq \argmin_{P_{X_1^T}: \sD_\TV(P_{X_1^T},Q_{X_1^T})\leq \epsilon} \sum_{k\in[K]}\bigg(\bigg(\sum_{x_1^T:f(x_1^T)=k}P_{X_1^T}(x_1^T)\bigg)- \alpha \bigg)_+,
    \end{align}
    and for any $x_1^T\in \calV^T$,
\begin{enumerate}[1),leftmargin=*]
    \item for all $\zeta_1^T$ s.t.\ $\sup_{\tilx_1^T\in \calB(f(x_1^T))}\gamma(\tilx_1^T,\zeta_1^T )=1$: $P_{\zeta_1^T|X_1^T}^{*,f}(\zeta_1^T|x_1^T)$ satisfies 
    \begin{align}
        \!\!\!\!\!\!\!\!\!\!\!\! \sum_{\tilx_1^T\in \calB_f(x_1^T)} \!\!\! P_{X_1^T}^{*,f}(\tilx_1^T)\sum_{\zeta_1^T}   P_{\zeta_1^T|X_1^T}^{*,f}(\zeta_1^T|\tilx_1^T) \sup_{\tilx_1^T\in \calB_f(x_1^T)}\gamma(\tilx_1^T,\zeta_1^T ) =\bigg(\sum_{\tilx_1^T\in \calB_f(x_1^T)}P_{X_1^T}^{*,f}(\tilx_1^T)\bigg)\wedge \alpha. 
    \end{align}
    
    \item $\forall \zeta_1^T$ s.t. $|\{x_1^T\in\calV^T:\gamma(x_1^T,\zeta_1^T)=1\}|=0$: $P_{X_1^T,\zeta_1^T}^{*,f}(x_1^T,\zeta_1^T)$ satisfies
    \begin{align}
         \!\!\!\!\!\!\!\!\!\sum_{\tilx_1^T\in \calB_f(x_1^T)} \!\!\! P_{X_1^T}^{*,f}(x_1^T)\sum_{\zeta_1^T:|\{x_1^T:\gamma(x_1^T,\zeta_1^T)=1\}|=0}   P_{\zeta_1^T|X_1^T}^{*,f}(\zeta_1^T|x_1^T) 
        =\bigg(\bigg(\sum_{\tilx_1^T\in \calB_f(x_1^T)}P_{X_1^T}^{*,f}(\tilx_1^T)\bigg)- \alpha \bigg)_+.\label{Eq: robust empty set prob}
    \end{align}

    \item all other cases of $\zeta_1^T$: $P_{X_1^T,\zeta_1^T}^{*,f}(x_1^T,\zeta_1^T)=0$.
\end{enumerate}
\end{theorem}

\begin{proof}[Proof of Theorem \ref{Thm: optimal robust detector}]
    When $f$ is an identity mapping, it is equivalent to Theorem \ref{Thm: optimal detector}. When $f:\calV^T\to [K]$ is some other function, following from the proof of Theorem \ref{Thm: optimal detector}, we consider three cases. 
   \setlist[itemize]{leftmargin=*}
   \begin{itemize}
       \item \underline{\textbf{Case 1:} $\calS\cap ([K]\cup \calV^T)\ne \emptyset$ but $|\calS|< K$.} It is impossible for the detector to detect all the watermarked text sequences. That is,  there exist $\tilx_1^T$ such that for all $\zeta_1^T$, $\gamma(\tilx_1^T,\zeta_1^T)=0$.
       Under this case,  in Appendix \ref{App: pf of Thm: robust min error}, $C(f(\tilx_1^T))=0\ne (\sum_{x_1^T:f(x_1^T)=f(\tilx_1^T)}P_{X_1^T}(x_1^T)) \wedge \alpha $, which means the $f$-robust Type-II error cannot reach the lower bound.

       \item \underline{\textbf{Case 2:} $\calS\cap ([K]\cup \calV^T)\ne \emptyset$ but $|\calS|= K$.} Under this condition, the detector needs to accept the form $\gamma(X_1^T,\zeta_1^T)= \mathbbm{1}\{f(X_1^T)=g(\zeta_1^T)\}$ so as to detect all possible watermarked text. Otherwise, it will degenerate to Case 1. We can see $f(X_1^T)$ as an input variable and rewrite the detector as $\gamma'(f(X_1^T),\zeta_1^T)=\gamma(X_1^T,\zeta_1^T)= \mathbbm{1}\{f(X_1^T)=g(\zeta_1^T)\}$. Similar the proof technique of Theorem \ref{Thm: optimal detector}, it can be shown that $C(k)$ in Appendix \ref{App: pf of Thm: robust min error} cannot equal $(\sum_{x_1^T:f(x_1^T)=k}P_{X_1^T}(x_1^T)) \wedge \alpha $ for all $k\in[K]$, while the worst-case $f$-robust Type-I error remains upper bounded by $\alpha$ for all $Q_{X_1^T}$ and $\epsilon$. 

       \item \underline{\textbf{Case 3:} Let $\Xi_\gamma(x_1^T)\coloneqq \{\zeta_1^T\in\calZ^T: \gamma(x_1^T,\zeta_1^T)=1\}$. $\exists x_1^T, y_1^T \in \calV^T, \text{ s.t. } f(x_1^T)\ne f(y_1^T)$} \\ \underline{and $\Xi_\gamma(x_1^T)\cap \Xi_\gamma(y_1^T)\ne\emptyset$.} For any detector $\gamma\notin \Gamma^*_f$ that does not belong to Cases 1 and 2, it belongs to Case 3. Let us start from a simple case where $T=1$, $\calV=\{x_1,x_2,x_3\}$, $K=2$, $\calZ=\{\zeta_1,\zeta_2,\zeta_3\}$, and $\calS=[2]$. Consider the mapping $f$ and the detector as follows: $f(x_1)=f(x_2)=1$, $f(x_3)=2$, $\gamma(x_1,\zeta_1)=\gamma(x_1,\zeta_1)=1$, $\gamma(x_3,\zeta_2)=1$, and $\gamma(x,\zeta)=0$ for all other pairs $(x,\zeta)$. When  $C(k)=(\sum_{x_1^T:f(x_1^T)=k}P_{X_1^T}(x_1^T)) \wedge \alpha$ for all $k\in[K]$, i.e.,
       \begin{align}
           &P_{X,\zeta}(x_1,\zeta_1)+P_{X,\zeta}(x_1,\zeta_2)+P_{X,\zeta}(x_2,\zeta_1)+P_{X,\zeta}(x_2,\zeta_2)=(P_X(x_1)+P_X(x_2))\wedge \alpha, \nn\\
           \text{and} \quad & P_{X,\zeta}(x_3,\zeta_2)=P_X(x_3)\wedge \alpha,\nn
       \end{align}
        then the worst-case $f$-robust Type-I error is lower bounded by
        \begin{align}
            &\max_{Q_{X_1^T}}\bbE_{Q_{X_1^T}\otimes P_{\zeta_1^T}}\left[\sup_{\tilx_1^T\in \calB_f(X_1^T)}\mathbbm{1}\{\gamma(\tilx_1^T,\zeta_1^T )=1\} \right]\\
            &\geq \bbE_{P_{\zeta_1^T}}\left[\sup_{\tilx_1^T\in \calB(1)}\mathbbm{1}\{\gamma(\tilx_1^T,\zeta_1^T )=1\} \right]\\
            &=(P_X(x_1)+P_X(x_2))\wedge \alpha+P_X(x_3)\wedge \alpha\\
            &>\alpha, \quad \text{if} \; P_X(x_1)+P_X(x_2)>\alpha \text{ or } P_X(x_3)>\alpha.
        \end{align}
         Thus, for any $Q_X$ such that $\{P_X\in\calP(\calV): \sD_\TV(P_X,Q_X)\leq \epsilon\}\subseteq \{P_X\in\calP(\calV): P_X(x_1)+P_X(x_2)>\alpha \text{ or } P_{X}(x_2)>\alpha\}$, the false-alarm constraint is violated  when $C(k)=(\sum_{x_1^T:f(x_1^T)=k}P_{X_1^T}(x_1^T)) \wedge \alpha$ for all $k\in[K]$.
         The result can be generalized to larger $(T,\calV,\calZ,K,\calS)$, other functions $f$, and other detectors that belong to Case 3.
   \end{itemize}
   In conclusion, if and only if $\gamma\in\Gamma^*$, the minimum Type-II error attained from  \eqref{Eq: opt-R} reaches the universal minimum $f$-robust Type-II error $\beta_1^*(f, Q_{X_1^T},\epsilon,\alpha)$ in \eqref{Eq: robustType-II LB} for all $Q_{X_1^T}\in\calP(\calV^T)$ and $\epsilon\in \bbR_{\geq 0}$.

Under the watermarking scheme $P_{X_1^T,\zeta_1^T}^{*,f}$,
the $f$-robust Type-I and Type-II errors are given by:

\textbf{$f$-robust Type-I error:}
\begin{align}
     \because & \forall y_1^T\in\calV^T, \quad  \bbE_{P_{\zeta_1^T}^{*,f}}\left[\sup_{\tilx_1^T\in \calB_f(y_1^T)}\mathbbm{1}\{\gamma(\tilx_1^T,\zeta_1^T )=1\} \right] \\ 
    &=\sum_{\zeta_1^T}\sum_{x_1^T}P_{X_1^T,\zeta_1^T}^{*,f}(x_1^T,\zeta_1^T) \sup_{\tilx_1^T\in \calB_f(y_1^T))}\mathbbm{1}\{\gamma(\tilx_1^T,\zeta_1^T )=1\}\\
    &=\sum_{x_1^T\in \calB_f(y_1^T)}P_{X_1^T}^{*,f}(x_1^T)\sum_{\zeta_1^T}   P_{\zeta_1^T|X_1^T}^{*,f}(\zeta_1^T|x_1^T)\sup_{\tilx_1^T\in \calB_f(y_1^T)}\mathbbm{1}\{\gamma(\tilx_1^T,\zeta_1^T )=1\}\\
    &=\bigg(\sum_{x_1^T\in \calB_f(y_1^T)}P_{X_1^T}^{*,f}(x_1^T)\bigg)\wedge \alpha \leq \alpha,\\
    & \!\!\!\! \text{and since any distribution $Q_{X_1^T}$ can be written as a linear combinations of $\delta_{y_1^T}$,} \nn\\
    \therefore &\sup_{Q_{X_1^T}}\bbE_{Q_{X_1^T}P_{\zeta_1^T}^{*,f}}\left[\sup_{\tilx_1^T\in \calB_f(X_1^T)}\mathbbm{1}\{\gamma(\tilx_1^T,\zeta_1^T )=1\} \right]\leq \alpha.
\end{align}
\textbf{$f$-robust Type-II error:}
\begin{align}
     \quad &1-\bbE_{P^{*,f}_{X_1^T,\zeta_1^T}}\left[\sup_{\tilx_1^T\in \calB_f(X_1^T)}\mathbbm{1}\{\gamma(\tilx_1^T,\zeta_1^T )=1\} \right] \nn\\
    &=1-\sum_{x_1^T}\sum_{\zeta_1^T}P_{X_1^T,\zeta_1^T}^{*,f}(x_1^T,\zeta_1^T) \sup_{\tilx_1^T\in \calB_f(x_1^T)}\mathbbm{1}\{\gamma(\tilx_1^T,\zeta_1^T )=1\}\\
    &=1-\sum_{k\in[K]}\sum_{x_1^T\in\calB(k)}P_{X_1^T}^{*,f}(x_1^T)\sum_{\zeta_1^T}   P_{\zeta_1^T|X_1^T}^{*,f}(\zeta_1^T|x_1^T)\sup_{\tilx_1^T\in \calB(k)}\mathbbm{1}\{\gamma(\tilx_1^T,\zeta_1^T )=1\} \\
    &=1-\sum_{k\in[K]}\bigg(\Big(\sum_{x_1^T\in\calB(k)}P_{X_1^T}^{*,f}(x_1^T)\Big)\wedge \alpha \bigg)\\
    &=\sum_{k\in[K]}\bigg(\Big(\sum_{x_1^T\in\calB(k)}P_{X_1^T}^{*,f}(x_1^T)\Big)- \alpha \bigg)_+.
\end{align}
The optimality of $P_{X_1^T,\zeta_1^T}^{*,f}$ is thus proved. 
\end{proof}

Figure \ref{Fig: uni type ii robust} compares the universally minimum Type-II errors with and without semantic-invariant text modification.
\begin{figure}[!t]
    \begin{center}
    \includegraphics[scale=0.8]{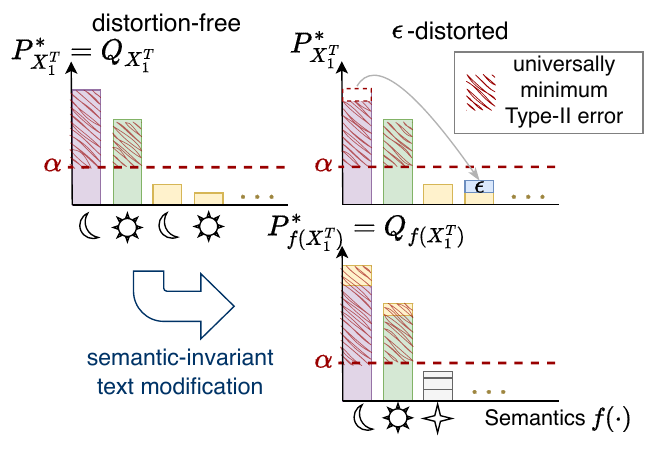}
    \vspace{-10pt}
    \caption{Universally minimum Type-II error w/o distortion and with semantic-invariant text modification.}
    \label{Fig: uni type ii robust}
    \end{center}
\end{figure}

\section{Broader Impacts}\label{App:impact}
This paper introduces a novel framework and algorithm for LLM watermarking, aimed at advancing the field of machine learning by enhancing AI safety and data authenticity. The primary positive impacts of our work include its potential to identify AI-generated misinformation, ensure integrity in academia and society, protect intellectual property (IP), and enhance public trust in AI technologies. However, given the dual-use potential of watermarking techniques, it is crucial to consider privacy concerns raised by possible misuse--such as unauthorized tracking of data. We encourage the development of ethical guidelines to ensure the responsible use and deployment of this technology. By considering both beneficial outcomes and potential risks, our work seeks to contribute responsibly to the machine learning community.





\newpage
\section*{NeurIPS Paper Checklist}

\begin{enumerate}

\item {\bf Claims}
    \item[] Question: Do the main claims made in the abstract and introduction accurately reflect the paper's contributions and scope?
    \item[] Answer: \answerYes{} 
    \item[] Justification: The abstract and introduction accurately summarize the paper's contributions, including the development of a unified theoretical framework for joint watermark scheme-detector optimization, derivation of optimal solutions, the DAWA algorithm, and experimental validation. These claims are consistent with the content presented in the subsequent sections and appendices.

\item {\bf Limitations}
    \item[] Question: Does the paper discuss the limitations of the work performed by the authors?
    \item[] Answer: \answerYes{} 
    \item[] Justification: The paper discusses ``Practical Challenges'' in implementing the theoretically optimal scheme (Section \ref{Sec: joint optimize}, page 5, line 225), which motivates the token-level design. Section \ref{Sec:extension} and Appendix \ref{App: ext to robust} discuss the robustness of DAWA, including its performance against certain attacks like deletion and paraphrasing (empirically addressed in Appendix \ref{App: add exp}, Table \ref{tab:more robust detection}) and outlines theoretical extensions for stronger future robustness, implying current trade-offs. A specific acknowledged limitation is that the sensitivity of DAWA to surrogate model mismatch in extreme cases was not empirically tested, although Appendix \ref{App: add exp} shows robustness to varied SLMs in non-extreme scenarios.

\item {\bf Theory assumptions and proofs}
    \item[] Question: For each theoretical result, does the paper provide the full set of assumptions and a complete (and correct) proof?
    \item[] Answer: \answerYes{} 
    \item[] Justification: Key theoretical results (Theorems \ref{Thm: universal min typeII}, \ref{Thm: optimal detector}, \ref{Thm: robust min error}, \ref{Thm: optimal robust detector}; Lemma \ref{Lem: token level errors}; Proposition \ref{Prop: robust against replace}) are presented with their assumptions stated. The paper indicates that complete proofs are provided in the corresponding appendices (e.g., Appendix \ref{App: pf of Thm: universal min typeII} for Theorem 1, Appendix \ref{App: pf of Thm: optimal detector} for Theorem 2, Appendix \ref{App: pf of Thm: robust min error} for Theorem \ref{Thm: robust min error}, Appendix \ref{App: optimal robust detector} for \ref{Thm: optimal robust detector}, Appendix \ref{App: pf of Lem: token level errors} for Lemma \ref{Lem: token level errors}), following standard practice.

    \item {\bf Experimental result reproducibility}
    \item[] Question: Does the paper fully disclose all the information needed to reproduce the main experimental results of the paper to the extent that it affects the main claims and/or conclusions of the paper (regardless of whether the code and data are provided or not)?
    \item[] Answer: \answerYes{} 
    \item[] Justification: Section \ref{Sec: experiments} (Experiments) and Appendix \ref{App: add exp} detail the experimental settings, including models used (Llama2-13B, Mistral-8$\times$7B, surrogate models), datasets (C4, ELI5, Wikipedia), prompts, DAWA parameters $(\eta,T)$, baselines, and evaluation metrics. This should provide sufficient information for others to understand and attempt to reproduce the core findings.

\item {\bf Open access to data and code}
    \item[] Question: Does the paper provide open access to the data and code, with sufficient instructions to faithfully reproduce the main experimental results, as described in supplemental material?
    \item[] Answer: \answerNo{} 
    \item[] Justification: We are committed to open science and will make the code for our DAWA algorithm and experimental scripts publicly available on  GitHub upon publication of this paper. The datasets used in our experiments (C4, ELI5, Wikipedia) are already publicly available and are cited in the manuscript. While an anonymized version of our code is not available at the time of submission, the paper provides detailed pseudo-code for DAWA (Appendix \ref{App: pseudo codes}) and a thorough description of our experimental setup (Section \ref{Sec: experiments} and Appendix \ref{App: add exp}) to allow for a detailed understanding and facilitate future reproduction.

\item {\bf Experimental setting/details}
    \item[] Question: Does the paper specify all the training and test details (e.g., data splits, hyperparameters, how they were chosen, type of optimizer, etc.) necessary to understand the results?
    \item[] Answer: \answerYes{} 
    \item[] Justification: Section \ref{Sec: experiments} ("Experiment Settings") and Appendix \ref{App: add exp} provide details on models, datasets, prompts, watermark parameters (e.g., $\eta=0.2,T=200$ for DAWA, strength for KGW+23), and evaluation metrics. The choice of $\eta$ is linked to theoretical considerations in Lemma \ref{Lem: token level errors} and Appendix \ref{App: pf of Lem: token level errors}.

\item {\bf Experiment statistical significance}
    \item[] Question: Does the paper report error bars suitably and correctly defined or other appropriate information about the statistical significance of the experiments?
    \item[] Answer: \answerNo{} 
    \item[] Justification: The paper reports primary evaluation metrics such as ROC-AUC and True Positive Rates (TPR) at specific False Positive Rates (FPRs), with detailed explanations of these metrics provided in Section \ref{Sec: experiments}. While error bars, confidence intervals, or formal statistical significance tests are not included for these reported values, our experiments are conducted on large, fixed test sets (e.g., $10^5$ texts for ultra-low FPR analysis, as described in Section  \ref{Sec: experiments}), which provides stability to the point estimates of these metrics. The observed performance differences, as shown in Figure \ref{fig:low_fpr} and Table \ref{tab:combined_detection}, are substantial. We acknowledge that quantifying variability, for instance, through multiple runs or bootstrapping, could further substantiate the comparisons, though it is often not standard practice for ROC-AUC and TPR reporting in this specific evaluation context due to the nature of these aggregate metrics on large test corpora and the computational demands of re-running experiments with large language models.

\item {\bf Experiments compute resources}
    \item[] Question: For each experiment, does the paper provide sufficient information on the computer resources (type of compute workers, memory, time of execution) needed to reproduce the experiments?
    \item[] Answer: \answerYes{} 
    \item[] Justification: Section \ref{Sec: experiments} mentions the use of ``Nvidia A100 GPUs''. Appendix \ref{App: add exp} (Table \ref{tab:gen_time_comparison}) provides average generation time comparisons for watermarked and unwatermarked text, giving an indication of the computational overhead of the watermarking process. This provides a baseline for understanding compute requirements.
    
\item {\bf Code of ethics}
    \item[] Question: Does the research conducted in the paper conform, in every respect, with the NeurIPS Code of Ethics \url{https://neurips.cc/public/EthicsGuidelines}?
    \item[] Answer: \answerYes{} 
    \item[] Justification: The research aims to develop methods for identifying AI-generated text, which can help mitigate misuse such as disinformation and plagiarism, aligning with ethical goals of responsible AI.  We have strived for integrity in our methodology, theoretical derivations, experimental procedures, and reporting of results. The datasets used are publicly available and appropriately cited, and the research did not involve direct human experimentation that would require IRB approval beyond standard dataset usage. However, given the dual-use potential of watermarking techniques, it is crucial to consider privacy concerns raised by possible misuse--such as unauthorized tracking of data. We encourage the development of ethical guidelines to ensure responsible use and deployment of this technology.


\item {\bf Broader impacts}
    \item[] Question: Does the paper discuss both potential positive societal impacts and negative societal impacts of the work performed?
    \item[] Answer: \answerYes{} 
    \item[] Justification:  The paper discusses societal impacts. Section \ref{Sec:intro}  outlines the significant positive societal impacts, such as mitigating AI-driven disinformation and enhancing data authenticity. The paper also acknowledges challenges related to watermark robustness and the evolving nature of attacks in Section \ref{Sec:extension} and Appendix \ref{App: ext to robust}. These inherent challenges in ensuring perfect, unbreakable watermarks implicitly relate to potential negative outcomes if the technology is misused or circumvented.
    The discussion in Appendix \ref{App:impact} also addresses the dual-use nature of watermarking, acknowledging potential negative impacts such as privacy concerns from misuse (e.g., unauthorized data tracking). In light of these risks, the paper underscores the importance of developing ethical guidelines for responsible deployment and use of this technology.

    
\item {\bf Safeguards}
    \item[] Question: Does the paper describe safeguards that have been put in place for responsible release of data or models that have a high risk for misuse (e.g., pretrained language models, image generators, or scraped datasets)?
    \item[] Answer: \answerNA{} 
    \item[] Justification: The paper introduces a watermarking methodology and an algorithm (DAWA). It does not release new large language models or novel datasets scraped from the internet that would pose a high risk for misuse. The models used for experiments (Llama2, Mistral) are existing models developed by other parties.

\item {\bf Licenses for existing assets}
    \item[] Question: Are the creators or original owners of assets (e.g., code, data, models), used in the paper, properly credited and are the license and terms of use explicitly mentioned and properly respected?
    \item[] Answer: \answerYes{} 
    \item[] Justification: The paper properly cites the original sources for datasets (C4 \citep{raffel2020exploring}, ELI5 \citep{fan2019eli5}, WMT19 dataset \citep{barrault2019findings}) and models (Llama2-13B \citep{touvron2023llama}, Mistral-8$\times$7B \citep{jiang2023mistral}, GPT-3 \citep{brown2020language}, MBART \citep{liu2020mbart}). The  licenses and terms of use for these assets are  explicitly mentioned (Appendix \ref{App: add exp}) and properly respected. 

\item {\bf New assets}
    \item[] Question: Are new assets introduced in the paper well documented and is the documentation provided alongside the assets?
    \item[] Answer: \answerNA{} 
    \item[] Justification: The paper introduces a new theoretical framework and algorithm (DAWA), which are documented within the paper itself (Appendix \ref{App: pseudo codes} for pseudo-code). No separate new datasets or software packages are being released that would require external documentation beyond the paper.

\item {\bf Crowdsourcing and research with human subjects}
    \item[] Question: For crowdsourcing experiments and research with human subjects, does the paper include the full text of instructions given to participants and screenshots, if applicable, as well as details about compensation (if any)? 
    \item[] Answer: \answerNA{} 
    \item[] Justification: The research does not involve crowdsourcing or direct human subject experiments (e.g., user studies with participant recruitment). Evaluation of text quality uses automated metrics (perplexity, BLEU) and existing datasets of human/AI text.

\item {\bf Institutional review board (IRB) approvals or equivalent for research with human subjects}
    \item[] Question: Does the paper describe potential risks incurred by study participants, whether such risks were disclosed to the subjects, and whether Institutional Review Board (IRB) approvals (or an equivalent approval/review based on the requirements of your country or institution) were obtained?
    \item[] Answer: \answerNA{} 
    \item[] Justification: As no direct human subject research or crowdsourcing was conducted (see Q14), IRB approval was not applicable.

\item {\bf Declaration of LLM usage}
    \item[] Question: Does the paper describe the usage of LLMs if it is an important, original, or non-standard component of the core methods in this research? Note that if the LLM is used only for writing, editing, or formatting purposes and does not impact the core methodology, scientific rigorousness, or originality of the research, declaration is not required.
    \item[] Answer: \answerYes{} 
    \item[] Justification: The proposed DAWA algorithm, specifically its detection phase, utilizes a Surrogate Language Model (SLM) as a component to approximate the watermarked distributions and enable model-agnostic detection (Section \ref{Sec: algorithm} ``Novel Tricks''); Algorithm \ref{algorithm1: detection}). This use of an LLM (the SLM) is integral to the methodology of the DAWA detector.

\end{enumerate}

\end{document}

%% file: preamble.tex
\usepackage[mathscr]{eucal}
\usepackage[cmex10]{amsmath}
\usepackage{epsfig,epsf,psfrag}
\usepackage{amssymb,amsthm,amsfonts,latexsym}
\usepackage{xcolor,url,overpic}
\usepackage{caption}
\usepackage{subcaption}
\usepackage{array}
\usepackage{verbatim}
\usepackage{bm}
\usepackage{bbm}
\usepackage{dsfont}
\usepackage{algorithmic}
\usepackage{algorithm}
\usepackage{verbatim}
\usepackage{textcomp}
\usepackage{mathrsfs}
\usepackage{epstopdf}
\usepackage{booktabs} 
\usepackage{thm-restate}
\usepackage{tikz}

\newcommand{\openone}{\leavevmode\hbox{\small1\normalsize\kern-.33em1}}

\catcode`~=11 \def\UrlSpecials{\do\~{\kern -.15em\lower .7ex\hbox{~}\kern .04em}} \catcode`~=13 

\allowdisplaybreaks[1]

\newcommand{\nn}{\nonumber}

\newcommand{\calA}{\mathcal{A}}
\newcommand{\calB}{\mathcal{B}}

\newcommand{\calG}{\mathcal{G}}

\newcommand{\calI}{\mathcal{I}}

\newcommand{\calO}{\mathcal{O}}
\newcommand{\calP}{\mathcal{P}}

\newcommand{\calR}{\mathcal{R}}
\newcommand{\calS}{\mathcal{S}}

\newcommand{\calV}{\mathcal{V}}

\newcommand{\calZ}{\mathcal{Z}}

\newcommand{\bv}{\mathbf{v}}

\newcommand{\rmH}{\mathrm{H}}

\newcommand{\bbE}{\mathbb{E}}

\newcommand{\bbN}{\mathbb{N}}

\newcommand{\bbR}{\mathbb{R}}

\DeclareMathAlphabet{\mathbsf}{OT1}{cmss}{bx}{n}
\DeclareMathAlphabet{\mathssf}{OT1}{cmss}{m}{sl}

\DeclareSymbolFont{bsfletters}{OT1}{cmss}{bx}{n}  
\DeclareSymbolFont{ssfletters}{OT1}{cmss}{m}{n}
\DeclareMathSymbol{\bsfGamma}{0}{bsfletters}{'000}
\DeclareMathSymbol{\ssfGamma}{0}{ssfletters}{'000}
\DeclareMathSymbol{\bsfDelta}{0}{bsfletters}{'001}
\DeclareMathSymbol{\ssfDelta}{0}{ssfletters}{'001}
\DeclareMathSymbol{\bsfTheta}{0}{bsfletters}{'002}
\DeclareMathSymbol{\ssfTheta}{0}{ssfletters}{'002}
\DeclareMathSymbol{\bsfLambda}{0}{bsfletters}{'003}
\DeclareMathSymbol{\ssfLambda}{0}{ssfletters}{'003}
\DeclareMathSymbol{\bsfXi}{0}{bsfletters}{'004}
\DeclareMathSymbol{\ssfXi}{0}{ssfletters}{'004}
\DeclareMathSymbol{\bsfPi}{0}{bsfletters}{'005}
\DeclareMathSymbol{\ssfPi}{0}{ssfletters}{'005}
\DeclareMathSymbol{\bsfSigma}{0}{bsfletters}{'006}
\DeclareMathSymbol{\ssfSigma}{0}{ssfletters}{'006}
\DeclareMathSymbol{\bsfUpsilon}{0}{bsfletters}{'007}
\DeclareMathSymbol{\ssfUpsilon}{0}{ssfletters}{'007}
\DeclareMathSymbol{\bsfPhi}{0}{bsfletters}{'010}
\DeclareMathSymbol{\ssfPhi}{0}{ssfletters}{'010}
\DeclareMathSymbol{\bsfPsi}{0}{bsfletters}{'011}
\DeclareMathSymbol{\ssfPsi}{0}{ssfletters}{'011}
\DeclareMathSymbol{\bsfOmega}{0}{bsfletters}{'012}
\DeclareMathSymbol{\ssfOmega}{0}{ssfletters}{'012}

\newcommand{\tilP}{\tilde{P}}

\newcommand{\tilQ}{\tilde{Q}}

\newcommand{\tilx}{\tilde{x}}

\DeclareMathOperator*{\argmax}{arg\,max}
\DeclareMathOperator*{\argmin}{arg\,min}

\newcommand{\Unif}{\mathrm{Unif}}

\usepackage{thmtools, thm-restate}
\newtheorem{theorem}{Theorem} 
\newtheorem{lemma}[theorem]{Lemma}

\newtheorem{proposition}[theorem]{Proposition}

\newtheorem{definition}{Definition} 
\newtheorem{example}{Example}